\RequirePackage[l2tabu, orthodox]{nag} 
\documentclass{article}

\usepackage[margin=1in]{geometry}
\usepackage{amsmath, amsthm, amssymb, graphicx, mathtools}
\mathtoolsset{showonlyrefs,showmanualtags}
\usepackage[dvipsnames]{xcolor}
\usepackage{url}
\usepackage{textcomp, gensymb}

\DeclareMathOperator*{\sig}{sig}

\newcommand{\N}{\mathbb{N}}
\newcommand{\R}{\mathbb{R}}
\renewcommand{\P}{\mathbb{P}}
\newcommand{\Q}{\mathbb{Q}}
\newcommand{\E}{\mathbb{E}}

\newcommand{\oggp}{(\Omega,\mathcal{G}, \mathbb{G}, \mathbb{P})}

\def\ind{{\mathchoice{1\mskip-4mu\mathrm l}{1\mskip-4mu\mathrm l}
		{1\mskip-4.5mu\mathrm l}{1\mskip-5mu\mathrm l}}}

\usepackage[normalem]{ulem}

\newtheorem{definition}{Definition}[section]
\newtheorem{proposition}[definition]{Proposition}
\newtheorem{corollary}[definition]{Corollary}
\newtheorem{remark}[definition]{Remark}
\newtheorem{example}[definition]{Example}
\newtheorem{lemma}[definition]{Lemma}
\newtheorem{theorem}[definition]{Theorem}

\usepackage{tikz}
\usetikzlibrary{svg.path}
\usepackage{pgfplots}
\usepgflibrary{fpu}
\usetikzlibrary{positioning,shadows,backgrounds}
\usepgfplotslibrary{fillbetween}

\definecolor{OIblack}{RGB}{0,0,0}
\definecolor{OIgreen}{RGB}{0,158,115}
\definecolor{OIblue}{RGB}{0,114,178}
\definecolor{OIbluelight}{RGB}{86,180,233}
\definecolor{OIyellow}{RGB}{240,228,66}
\definecolor{OIorange}{RGB}{230,159,0}
\definecolor{OIred}{RGB}{213,94,0}
\definecolor{OIpink}{RGB}{204,121,167}
\definecolor{copper}{RGB}{184, 115, 51}
\usepackage{wrapfig}

\pgfplotsset{compat=1.18}

\title{Understanding the Commodity Futures Term Structure\\ Through Signatures}

\author{Hari P. Krishnan\thanks{ %
H. P. Krishnan, SCT Capital Management, 
	(e-mail: \texttt{harifinancialmarkets@gmail.com}).}
\and
Stephan Sturm\thanks{ %
S. Sturm, Department of Mathematical Sciences, Worcester Polytechnic Institute, 100 Institute Road, Worcester MA 06109, USA
	(e-mail: \texttt{ssturm@wpi.edu}). Corresponding author. ORCID: \texttt{https://orcid.org/0000-0002-7214-8843}}
}

\begin{document}

\maketitle

\begin{abstract}
Signature methods have been widely and effectively used as a tool for feature extraction in statistical learning methods, notably in mathematical finance. They lack, however, interpretability: in the general case, it is unclear why signatures actually work. The present article aims to address this issue directly, by introducing and developing the concept of signature perturbations. In particular, we construct a regular perturbation of the signature of the term structure of log prices for various commodities, in terms of the convenience yield. Our perturbation expansion and rigorous convergence estimates help explain the success of signature-based classification of commodities markets according to their term structure, with the volatility of the convenience yield as the major discriminant. 
\end{abstract}

\vspace{5mm}
 
\begin{flushleft}
\end{flushleft}

\section{Introduction}

Most commodities share a common property: they are consumed as food or as inputs to production. Many of the largest markets, such as crude oil and copper, have well-developed futures markets that allow producers and end users to hedge their revenues and costs, respectively. While there are hedgers and speculators in all of these markets, there are significant differences in futures curve dynamics across commodities. This can lead to serious difficulties when pooling futures price and return data across markets, in an effort to build models that are cross-sectionally stable. Non-storable commodities, such as electricity, pose a particular challenge, given their highly variable term structure dynamics and extreme event risk. Trading models that work for a highly storable commodity, such as copper, may not be suitable for electricity, and \textit{vice versa}. 

As a result, the development of classification models for various commodities is a worthwhile objective. These models generally require some feature engineering. Standard  features, such as volatility, skewness and the standard deviation of calendar spreads along the term structure, are reasonable if somewhat obvious choices. However, these statistics do not take into account the order in which returns are realized, hence ignore a great deal of structure in cross-sectional price movements. 

Path signatures, originally developed in \cite{Che57} and popularized by the development of rough analysis starting with \cite{L94, L98}, offer a natural alternative to the moments of the historical distribution. Hambly and Lyons \cite{HL10} have shown that the signature vector gives a nearly complete characterization of a path, in highly compressed form. Significantly for our purposes, this result has been generalized to (Stratonovich) diffusions in \cite{GQ16}. \cite{Gra13} has demonstrated the success of signature features in classifying handwritten characters. Notable early applications of the signature method to finance include \cite{GLKF14, LNO14}. Finally, a modern application pipeline is discussed in \cite{MFKL21}.

The futures term structure of a given commodity contains the prices (or equivalently, log prices) of contracts with different times to maturity. As the term structure evolves over time, it generates a multi-dimensional path, with an associated signature vector. The authors have found that, using signature features, accurate classifiers can be built that distinguish between metals and grains, or grains and softs. Significantly, these results are retained even after normalizing volatility across markets \cite{IKMSS25}.

This leads to a crucial methodological question, namely how term structure variability impacts the signature vector of a given commodity. Characterizing this dependence is the goal of this paper. Understanding the relationship between term structure dynamics and signatures provides crucial insight into how a signature-based classifier is able to distinguish between different categories of commodities. This requires combining the model-free path signature transformation with more traditional model-specific approaches found in quantitative finance. Specifically, the (truncated) signature feature set in a classification algorithm essentially acts as a \textit{black box}. The values in a given signature vector are extremely difficult to interpret in isolation. However, we can use a parameterized model of term structure dynamics for commodity futures to build our understanding. The model generates paths, where each path has a corresponding, theoretical signature. As we vary the parameters in the model, we can estimate how the signature vector changes. The parameters that have the greatest impact on the signature are likely to be hidden features where our classifier is focusing its attention. In this way, we turn the feature set from a black box into an interpretable quantity dependent on model parameters. We can then interpret the signature in market terms, relative to a standard financial model. The importance of interpretability in a feature set is detailed in \cite{ZALBEV22}. In their framework, (truncated) signatures are a \textit{predictive}, \textit{model compatible} feature set, where the perturbation acts as an \textit{interpretable transform} to provide a \textit{meaningful} and \textit{trackable} outcome.

Standard models for describing term structure dynamics generally involve a quantity known as the commodity convenience yield, as in \cite{Gem05}. The convenience yield has a large bearing on the shape of the futures term structure. For example, high short-term convenience yields indicate a demand for immediacy and cause the prices of near-term futures contracts to rise, relative to longer-term maturities. More generally, when the convenience yield for a given commodity is volatile, we would expect significant variability in the spread between contracts with different maturities. Intuitively, this suggests a relationship between the volatility of convenience yields and the cross-sectional variance across signatures of a given order.

Without reference to a stochastic model, Kaldor \cite{Kal39, Kal40} introduced and refined the notion of the convenience yield. A modern discussion of the concept can be found in \cite{LauND}. We accept that the very notion of a convenience yield is quite contentious for certain commodities markets. Geman \cite{Gem05} questions its relevance for electricity markets, where the physical commodity is non-storable. However, others, such as Carmona and Ludkovski \cite{CL04} happily invoke the convenience yield and focus more on model formulation.

The most commonly used convenience yield model for commodity prices is likely Gibson and Schwartz \cite{GS90} (see also \cite{Sch97}). Here, futures prices are described by the joint stochastic evolution of the spot price and convenience yield over time. The convenience yield follows a mean-reverting, Ornstein--Uhlenbeck process, where the (geometric) drift of the spot price is given by the short rate minus the convenience yield. Assuming a constant risk-free rate and annualized cost of storage, we can then price a futures contract of any given maturity. While the Gibson--Schwartz model has been criticized as overly restrictive with regard to possible term structure dynamics (e.g., \cite{CL04}), it is a useful baseline model that strikes a balance between simplicity and accuracy and in this respect plays a similar role to Black--Scholes for equities.

From a mathematical standpoint, a crucial question is how term structure variability in a given model impacts the signature vector. When the futures curve only depends on the spot price and time to maturity, the futures path signature reduces to a deterministic function of the signature of the spot price. Given a fixed amount of variability in the shape of the term structure, we want to estimate how much the signature will change. Assuming that our estimates are sharp enough, we can then relate signature variability to a measure of uncertainty in term structure dynamics. This line of reasoning provides an explanation as to why signatures are suitable features in a classification algorithm and provides a way to interpret signatures in term of model parameters. 

In particular, we want to apply the idea of multiscale perturbations in partial differential equations (PDEs) to the study of signature approximations. We were inspired by and have relied upon the theory of multiscale stochastic volatility models, as in \cite{FPSS11}. These stochastic models were originally developed for the pricing of contingent claims in terms of slow regular and fast singular volatility perturbations. In this paper, we will perturb the commodity convenience yield rather than volatility, using the same mathematical machinery. Our goal is to characterize term structure dynamics using a quantity that is analogous to implied volatility, in the sense that it is not directly observed in the market. Mimicking the approach of \cite{FPSS11}, we can construct a slow regular perturbation of the convenience yield process within the Gibson--Schwartz model and analyze its impact on the signature of the futures term structure. Here, we restrict to "slow" regular perturbations of the term structure and prove convergence of the signature of the term structure path, under appropriate conditions.

Qualitatively speaking, signatures in general and the perturbation approach for interpretability are powerful ways to characterize the term structure dynamics of various commodities. For certain commodities, such as copper, contracts with different maturities tend to move in lock step. The signature of the term structure for such a commodity can be thought of as a small perturbation of the signature of spot prices. By contrast, other commodities have more variable futures curve dynamics. We would expect higher cross-sectional variance across signatures of a given order in this case. The perturbation approach allows us to understand the commodities with a more dynamic futures curves as perturbations of the more regular ones.

In particular this approach allows us to isolate the parameters of the convenience yield process that have the largest impact on the signature expansion. We find that the first order term is depending solely on the volatility of the convenience yield. This theoretical finding can be backed up by the empirical literature on the topic, e.g., \cite{PSSW23} and \cite{KMPV18} which show that the volatility of the convenience yield is for non-storable commodities such as natural gas is by one to two orders of magnitude larger than that of metals.

The remainder of the paper is structured as follows. Section \ref{sec:main} explains our model and the concept of path signatures. We then develop an approximation scheme of the model and show how this leads to an approximation of the path signature. Specifically, Theorem \ref{thm:exp-first} explains the main result on a high level which allows us to interpret it in financial terms and also confirm the validity by comparing the conclusions to the empirical literature. Section \ref{sec:approx} contains the mathematically rigorous deviation of the result. We introduce a convenient notion of convergence of signature by defining the weighted signature norm on the ambient space of sequences of random variables and derive sharp constants for the approximation of iterated integrals and signatures. More technical aspects of the proof are relegated to Appendix \ref{sec:proofs}. The concluding Section \ref{sec:conclusion} summarizes and contextualizes the main findings of the paper.

\section{Setting and Result}\label{sec:main}

We consider a filtered probability space $\oggp$ with complete and right-continuous filtration $\mathbb{G} = (\mathcal{G}_t)_{0 \leq t \leq 1}$ assuming $\mathcal{G}_1 = \mathcal{G}$. The commodity market consisting of a risky asset with spot price modeled by a continuous $(\mathbb{G}, \mathbb{P})$-semimartingale $\tilde{X}$ and a riskless asset with continuously compounded, constant interest rate $r$. To exclude arbitrage in the sense of \textit{no free lunch with vanishing risk} (see \cite{DS06}), we require the existence of a probability measure $\Q$ under which the discounted asset process $(\tilde{X}_te^{-r t})_{t\geq 0}$ is a local $\mathbb{G}$-martingale. 

Several futures contracts with different maturities have been written on the underlying asset $X$. To simplify the model, we assume that the futures are continuously rolled. This approximates actual market practice. Adapting Musiela parametrization, we use $T_1, \, T_2, \ldots T_d$ to denote the times to maturity of the futures contracts $\tilde{F}^1, \, \tilde{F}^2, \ldots \tilde{F}^d$. Using a standard short-rate model, we assume for the time-$t$-price of the futures contact $\tilde{F}^k$, $k \in \{1, \cdots,d\}$,
\[
F^k_t = \tilde{X}_t \E^\Q\Bigl[e^{\int_{t}^{t+T_k} r_u + s_u - C_u \, du} \, \Big\vert \, \mathcal{G}_t \Bigr],
\]
where $r$ is the instantaneous interest rate, $s$ is the annualized percentage storage cost, $C$ the convenience yield and $\Q$ a risk-neutral measure. In the general case, $r$, $s$ and $C$ can all be time-dependent and stochastic. We write $\tilde{\mathbb{F}} = (\tilde{F}^1, \ldots, \tilde{F}^d)$ to describe the full term structure. The corresponding return process are $X = \log{(\tilde{X})}$ for the spot and $F^k = \log{(\tilde{F}^k)}$, $k \in \{1, \cdots,d\}$, for the futures as well as $\tilde{\mathbb{F}} = (\tilde{F}^1, \ldots, \tilde{F}^d)$ for the term structure of the returns. The object of interest is the forward (return) curve, i.e., the path $\mathbb{F}\, : \, [0,1] \to \R^d$.

The essential information of the forward curve is codified in its signature, the collection of iterated integrals with respect tu the futures return processes. More specificallly, the \textit{signature of $\mathbb{F}$} is given by
\[
S_{0,1}(\mathbb{F}) = \Bigl(1, S_{0,1}^1 (\mathbb{F}), \ldots, S_{0,1}^d (\mathbb{F}), S_{0,1}^{1,1} (\mathbb{F}), S_{0,1}^{1,2} (\mathbb{F}), \ldots \Bigr)
\]
where
\[
S^{i_1, \ldots, i_k}_{0,1}(\mathbb{F}) = \int_0^1 \int_0^{t_k} \cdots \int_0^{t_3}\int_0^{t_2} 1 \, \circ dF^{i_1}_{t_1}\circ dF^{i_2}_{t_2} \cdots \circ dF^{i_{k-1}}_{t_{k-1}} \circ dF^{i_k}_{t_k}
\]
for $i_1, \ldots, i_k \in \{1,\ldots d\}$. The iterated integrals are here stochastic integrals in the sense of (Fisk--)Stratonovich (indicated by $\circ $). Similarly, we have for the (one-dimensional) spot return $X$ the signature
\[
S_{0,1}(X) = \Bigl(1, S_{0,1}^1 (X), S_{0,1}^{1,1} (X), S_{0,1}^{1,1,1} (X), \ldots \Bigr)
\]

We start with a simple example, namely the case where interest, storage costs and convenience yield are constant. While much too simplistic be deployed in practice by itself, this model will be the effective starting point of our perturbation analysis.
\begin{example}
Let's assume that interest rate, storage rate and convenience yield are constant, than we have
\[
\tilde{F}^k_t = e^{(r + s - C)T_k}\tilde{X}_t,
\]
and hence for the returns
\[
F^k_t = (r + s - C)T_k + X_t.
\]
It follows thus for the signature that
\begin{align*}
S^{i_1, \ldots, i_k}_{0,1}(\mathbb{F}) &= \int_0^1 \int_0^{t_k} \cdots \int_0^{t_3}\int_0^{t_2} 1 \, \circ dF^{i_1}_{t_1}\circ dF^{i_2}_{t_2} \cdots \circ  dF^{i_{k-1}}_{t_{k-1}} \circ dF^{i_k}_{t_k}\\
&= \int_0^1 \int_0^{t_k} \cdots \int_0^{t_3}\int_0^{t_2} 1 \, \circ  dX^{i_1}_{t_1} \circ dX^{i_2}_{t_2} \cdots \circ  dX^{i_{k-1}}_{t_{k-1}} \circ  dX^{i_k}_{t_k} = S^{1, \ldots, 1}_{0,1}(X) 
\end{align*}
as the integrals with respect to constants are all zero. Thus, while the signature is still a vector of random variables, all signatures terms of the same order (i.e., with the same number of iterated integrals), are identical. This implies that all information about the variability of the futures curve is encoded in spot dynamics.
\end{example}

\subsection{Perturbing the Gibson--Schwartz model}

The clean results for the constant case suggest that a perturbation approach might be helpful. We are thinking about a multiscale approach in the spirit of Fouque--Papanicolaou--Sircar--S{\o}lna \cite{FPSS11}, focusing first on the slow scale. To do so, we need a more specific model. A natural choice the the Gibson--Schwartz model as originally proposed in \cite{GS90}, for a discussion about context and limitations see \cite{Sch97} and \cite{CC14}: 

Under the (fixed) risk-neutral measure $\Q$ we have for the spot prices
\begin{alignat*}{2}
d\tilde{X}_t &= \bigl(r - C_t\bigr) \tilde{X}_t \, dt + \sigma \tilde{X}_t \, dW_t^1, \qquad &&\tilde{X}_0 = e^x, \\
dC_t &= \kappa\bigl(\theta - C_t\bigr) \, dt + \gamma \, dW_t^2, \qquad && C_0 = c,
\end{alignat*}
where $W^1$, $W^2$ are two Brownian motions with constant correlation $\rho \in (-1,1)$. We assume that the filtration $\mathbb{G}$ is the (augmented) natural filtration of the Brownian motion $(W^1, W^2)$. As we are working with returns and also need Stratonovich integration we convert this to
\begin{alignat*}{2}
dX_t &= \Bigl(r - \frac{\sigma^2}{2} -C_t\Bigr) \, dt + \sigma \, \circ  dW_t^1, \qquad && X_0 = x, \\
dC_t &= \kappa\bigl(\theta - C_t\bigr) \, dt + \gamma \, \circ  dW_t^2, \qquad && C_0 = c.
\end{alignat*}
We recognize that for these SDEs the It\^{o} and Stratonovich formulations agree (as the integrands are deterministic), but we we will need the Stratonovich formulation for the further work. Assuming interest rate and storage costs to be constant, futures prices can be expressed via spot price and convenience yield processes.

\begin{proposition}\label{prop:GS}
The signature terms of the futures returns term structure can be expressed in terms of the Gibson--Schwartz processes as
\begin{align*}
S^{i_1, \ldots, i_k}_{0,1}\bigl(\mathbb{F}\bigr) = \sum_{(j_1, \ldots, j_k) \in \{1,2\}^k} \prod_{l=1}^k B(T_{i_l})^{j_l-1} S^{j_1, \ldots, j_k}_{0,1}\bigl((X,C)\bigr) 
\end{align*}
with 
\[
B\bigl(T_k\bigr) =  \frac{1}{\kappa}\Bigl(1-e^{-\kappa T_k}\Bigr).
\]
\end{proposition}
We note that we reduce the term structure signature to iterated integrals of the two-dimensional path $(X,C)$, with $X$ here referred to as first component and $C$ as the second. Thus  for each integral w.r.t. $C$ there is a corresponding $B(T_{i_l})$-factor in front of the iterated integral.

\begin{proof}
\begin{align*}
\tilde{F}^k_t & = X_t \E^\Q\Bigl[e^{\int_{t}^{t+T_k} r + s - C_u \, du} \, \Big\vert \, \mathcal{G}_t \Bigr]\\
& = e^{(r + s)T_k} X_t \E^\Q\Bigl[e^{-\int_{t}^{t+T_k} C_u \, du} \, \Big\vert \, \mathcal{G}_t \Bigr]\\
& = e^{(r + s)T_k} X_t A(T_k) e^{- B(T_k) c_t}
\end{align*}
from the Va\v{s}\'{i}\v{c}ek bond price formula (see, e.g., \cite[Section 2.6]{JYC09}), where
\begin{align*}
A\bigl(T_k\bigr) & = \exp\Biggr(\biggr(\theta - \frac{\gamma^2}{2 \kappa^2}\biggr)\bigl(B(T_k) - T_k\bigr) -\frac{\gamma^2}{4\kappa}B(T_k)^2\Biggr)\\
B\bigl(T_k\bigr) & =  \frac{1}{\kappa}\Bigl(1-e^{-\kappa T_k}\Bigr).
\end{align*}
Switching now to returns we get
\[
F^k_t = X_t - B\bigl(T_k\bigr)C_t +K
\]
where $K$ is a constant that doesn't matter for the signatures, as we care only about differentials
\[
dF^k_t = dX_t - B\bigl(T_k\bigr)dC_t.
\]
\end{proof}

We can also consider the Gibson--Schwartz model as a perturbation around a model with constant convenience yield. In practical terms, this corresponds to the case where term structure dynamics are nearly deterministic. For this we introduce the parameter $\delta>0$ that can, thanks to Brownian scaling, be considered as a (slow) speed of the convenience yield modeled,
\begin{alignat*}{2}
dX^\delta_t &= \bigl(r - \frac{\sigma^2}{2} - C^\delta_t\bigr) \, dt + \sigma \, dW_t^1, \qquad && X^\delta_0 = x, \\
dC^\delta_t &= \delta \kappa\bigl(\theta - C^\delta_t\bigr) \, dt + \sqrt{\delta}\gamma \, dW_t^2 , \qquad && C^\delta_0 = c.
\end{alignat*}

If $\delta=1$ we recover the original Gibson--Schwartz model. In the case of small $\delta \ll 1$ we can think about the convenience yield moving at a glacial pace relative to the spot price dynamics. In the limiting case where $\delta = 0$, $C$ remains constant at $c$ over time.

The formula of Proposition \ref{prop:GS} becomes
\begin{align}
S^{i_1, \ldots, i_k}_{0,1}\bigl(\mathbb{F}^\delta\bigr) &= \sum_{(j_1, \ldots, j_k) \in \{1,2\}^k} \prod_{l=1}^k B^\delta(T_{i_l})^{j_l-1} S^{j_1, \ldots, j_k}_{0,1}\Bigl(\bigl(X^\delta,C^\delta\bigr)\Bigr)\nonumber \\
B^\delta\bigl(T_k\bigr)  &=  \frac{1}{\delta \kappa}\Bigl(1-e^{-\delta \kappa T_k}\Bigr)\label{eq:GS-exp}
\end{align}
where the power series converges everywhere on $\mathbb{R}$. We are interested in writing this representation in an expansion of powers of $\delta$. 
As
\[
B^\delta\bigl(T_k\bigr)  =  \frac{1}{\delta \kappa}\Bigl(1-e^{-\delta \kappa T_k}\Bigr) = T_k\sum_{j=0}^\infty \frac{(-\kappa T_k)^j}{(j+1)!} \delta^j
\]
we note that most straightforwardly that $n$-th order approximation of $B$ is given by
\[
B^{(n)}(T_k) := T_k\sum_{j=0}^n \frac{(-\kappa T_k)^j}{(j+1)!} \delta^j.
\]
For the perturbed convenience yield processes $C$ we will use the approximations,
\[
C^{(n)} = \sum_{j=0}^n \sqrt{\delta}^j \hat{C}^{(j)}, \qquad  \hat{C}^{(j)}_t = \left\{\begin{array}{ll} c & \text{ if } j= 0,\\  \gamma \frac{(-\kappa)^k}{k!} \int_0^t (t-s)^k \, dW^2_s  & \text{ if } j = 2k+1, k \in \N, \\  (c-\theta)\frac{(-\kappa)^{k+1}}{(k+1)!}t^{k+1}  & \text{ if } j =2k+2, k \in \N, \end{array}\right. \]
while for the spot returns we use
\[
X^{(n)} = \sum_{j=0}^n \sqrt{\delta}^j\hat{X}^{(j)} \qquad \hat{X}^{(j)}_t = \left\{\begin{array}{ll} x_0 - \bigl(c- r+\frac{\sigma^2}{2}\bigr) t + \sigma W_t^1 & \text{ if } j= 0,\\ \gamma \frac{(-\kappa)^k}{(k+1)!} \int_0^t (t-s)^{k+1} \, dW^2_s  & \text{ if } j = 2k+1, k \in \N, \\\  -(c-\theta)\frac{(-\kappa)^{k+1}}{(k+2)!}t^{k+2}  & \text{ if } j =2k+2, k \in \N. \end{array}\right.
\]
We will provide more details on the accuracy of this expressions in the next section, cf. Proposition \ref{prop:GS-approx}.

\begin{remark}\label{rem:const}
We note that the approximation for the perturbed convenience yield $C^\delta$ have been derived by using separate power series for the deterministic and stochastic part of the closed-form representation of the Ornstein--Uhlenbeck process
\[
C^\delta_t = ce^{-\delta \kappa t} + \theta \bigl(1-e^{- \delta\kappa  t}\bigr) + \sqrt{\delta} \gamma \int_0^t e^{-\delta \kappa(t-s)} \, dW^2_s.
\]
The naive approach, often used for PDEs, to just expand the process in a power series and then compare the coefficients of the SDE fails in this case, as it reconstructs only the deterministic part, but not the stochastic one. For the perturbed returns we use just the explicit formula given the first order term of the convenience yield to get
\[
\hat{X}^{(0)}_t = x_0 - \Bigl(c-r+ \frac{\sigma^2}{2}\Bigr) t + \sigma W_t^1,
\]
while for the higher order terms we use the integration of the convenience yield terms
\[
\hat{X}^{(n)}_t = - \int_0^t \hat{C}^{(n)}_s \, ds,
\]
whence
\begin{equation}\label{eq:Xn-rep}
X^{(n)}_t = x_0 + \int_0^t \bigl(r - \frac{\sigma^2}{2} - C_s^{(n)}\bigr) \, ds + \sigma W^1_t.
\end{equation}
\end{remark}

Finally, we are writing the approximation of the futures term structure perturbation \eqref{eq:GS-exp} as
\begin{align}\label{eq:GS-approx} 
S^{i_1, \ldots, i_k}_{0,1}\bigl(\mathbb{F}^{(n)}\bigr) &= \sum_{(j_1, \ldots, j_k) \in \{1,2\}^k} \prod_{l=1}^k B^{(\lfloor \frac{n}{2}\rfloor)}(T_{i_l})^{j_l-1} S^{j_1, \ldots, j_k}_{0,1}\bigl((X^{(n)},C^{(n)})\bigr) 
\end{align}
where we use
\[
B^{(n)}(T_k) := \sum_{j=0}^n (-\kappa)^j \frac{T_k^{j+1}}{(j+1)!}\delta^j.
\]

We claim now that we can write the signature of the futures term structure as an expansion around the time-scale parameter $\delta$. In particular we will prove the following theorem.
\begin{theorem}\label{thm:exp-first}
The signature $S_{0,1}\bigl(\mathbb{F}^\delta\bigr)$ of the futures can be written in the expansion 
\begin{align}\label{eq:main-approx}
S_{0,1}\bigl(\mathbb{F}^\delta\bigr) & =E_0 +
 \Bigl(E_0 + \gamma E_{1,1}\Bigr) \sqrt{\delta} +\Bigl(E_0 + \gamma^k E_{2,1} + \kappa(\theta-c)E_{2,2}\Bigr)\delta  \nonumber \biggr.\\
&\phantom{==} + \biggl. \Bigl(E_0 + \gamma^k E_{3,1} + \bigl((\theta-c)\kappa\bigr)^k E_{3,2} + \gamma \kappa E_{3,3} + \gamma^{k-1}(\theta-c)\kappa E_{3,4} +\gamma\bigl((\theta-c)\kappa\bigr)^{k-1}E_{3,5}\Bigr)\sqrt{\delta}^3  \biggr.\nonumber \\
&\phantom{==} \biggl.  +  \mathcal{O}\bigl(\delta^2\bigr).
\end{align}
where the $E$-terms are random variables that do not depend on the model parameters $\gamma$, $\kappa$ and $\theta$ (but on $c$, the spot parameters and the maturities of the futures contracts and the interest rate).
\end{theorem}
Of course this formula and in particular the meaning of the Big-Oh notation will have to be clarified in detail, as the setting is stochastic and signature and the $E$-terms are random variables, but we postpone this to the next section \ref{sec:approx}. For now we focus on the implication of this approximation.

\subsection{Interpretation of the Results}

The zeroth order term depends solely on the spot price, spot volatility and the initial level $c$ of the convenience yield (thus we are in the in the case of constant convenience yield discussed in Remark \ref{rem:const} and have that result as the baseline of our expansion). The  first order term (i.e., the term of order $\sqrt{\delta}$) depends additionally on the volatility $\gamma$ of the convenience yield (through the representation of $C^{(1)}$), the reversion speed $\kappa$ and mean reversion level $\theta$ only appear in the second order term (i.e., of order $\delta$), and they appear only jointly in the form $\kappa(\theta - c)$. A possibility to separate their effect comes only with the third order term (i.e., the order $\sqrt{\delta}^3$) where $\kappa$ appears isolated in contributing terms. None of the terms depends on storage costs (as they were assumed constant). Overall, the strongest effect of the model parameters (besides the initial convenience yield) has by $\gamma$, the volatility of the convenience yield. Thus if we employ a classification algorithm on the feature set given by tthe (truncated) signature vecor, it is most likely to pick up on differences in the volatility of the convenience yield; this is what drives the superior performance of the signature-based method in discerning storable from non-storable commodities..

This is in line with observation of empirical work that studies the volatility of the convenience yield, notably Prokopczuk, L. Symeonidis, C. Wese Simen, R. Wichmann \cite[Table 1]{PSSW23} and Koijen, Moskowitz, Pedersen, and Vrugt \cite[Table 1]{KMPV18} who study the closely related concept of carry (which coincides for commodities essentially with the convenience yield up to a factor). Both papers show that there is a huge variety in the standard deviation of the convenience yield with very storable commodities as metals have the smallest and very non-storable commodities as natural gas the highest (typically more than ten times higher). There is some difference between the two papers, not only that they track slightly different quantities (besides the difference between carry and convenience yield, \cite{PSSW23} focus on the convenience yield between the first two futures contracts while \cite{KMPV18} considers that from spot to first futures) but they use time series of different length for estimation. This might contribute to a good part to the difference, as the convenience yield volatility is highly heteroscedastic (see \cite{LT11} for a general analysis and \cite{CWZ06} for the example of the natural gas market). Additionally, the recent financialization of commodity markets also affected the structure of the convenience yield, e.g., \cite{MPh24} show that it became considerably negative correlated with the VIX.

We can also illustrate the results by plotting the path of the first two futures prices for different commodities, see Figure \ref{fig:paths}. We see that for a highly storable commodity, the futures prices exhibit some volatility, but the different contracts move in lock-step, creating nearly a $45\degree$ line. Natural gas, on the other hand, does not only exhibit a higher volatility, but shows significant differences in one-day return movements between the different contracts. Softs as a class of commodities show the widest degree of inter-class variability. The figure shows the difference between cocoa (volatile, but very lock-step behavior, the lowest standard deviation of any convenience yield of softs in \cite{PSSW23} and \cite{KMPV18}) and cotton (idiosyncratic moves of the different contracts, the highest standard deviation of convenience yield of any soft in \cite{PSSW23} and \cite{KMPV18}). This also explains why the signature based classification algorithm performs not as well when asked to classify \textit{softs vs. grains}: they do not differ clearly on the standard deviation of the convenience yield, as \textit{grains vs. metals} or \textit{storables vs. non-storables} do.

\begin{figure}[htb]
\centering
	\begin{tikzpicture}[scale=0.55]
		\node[black] at (5.5,6) {\textbf{July 2023}};
\begin{scope}[scale = 0.2]
	\draw[gray!80,->] (-1,0) -- (21,0) node at (23,0) {\footnotesize $x$};
	\draw[gray!80,->] (0,-1) -- (0,21) node at (0,23) {\footnotesize $y$};
	\draw[gray!80,] (20,-0.2) -- (20,0.2) node at (20.3,-2.2) {\footnotesize$6\%$};
	\draw[gray!80,] (-0.2,20) -- (0.2,20) node at (-2.4,20.2) {\footnotesize $6\%$};
	\node[copper] at (15,22) {\small \textbf{Copper}};
	\begin{scope}[scale=0.33]
		\draw[copper] (-10.61,-10.04) -- (-7.24,-7.11) -- (13.29,13.11) -- (-0.57,-0.12) -- (-6.26,-6.10) -- (21.45, 21.18) --
(22.47, 22.63) -- (-3.09,-2.62) -- (-23.06, -21.89) -- (-0.24, -1.89) -- (-4.95, -5.01) -- (6.72, 6.74) -- (-4.41, -4.20) -- (7.64, 7.63) -- (17.95, 17.99) -- (-7.12, -6.82) -- (-5.99, -5.81) -- (11.13, 11.06) -- (19.42, 19.39);
\end{scope}
\end{scope}
\begin{scope}[scale = 0.2, yshift =-40cm]
	\draw[gray!80,->] (-1,0) -- (21,0) node at (23,0) {\footnotesize $x$};
	\draw[gray!80,->] (0,-1) -- (0,21) node at (0,23) {\footnotesize $y$};
	\draw[gray!80,] (20,-0.2) -- (20,0.2) node at (20.3,-2.2) {\footnotesize$6\%$};
	\draw[gray!80,] (-0.2,20) -- (0.2,20) node at (-2.4,20.2) {\footnotesize $6\%$};
	\node[OIblack] at (15,22) {\textbf{Cocoa}};
	\begin{scope}[scale=0.33]
		\draw[OIblack] (-9.00, -8.90) -- (-17.39, -16.90) -- (-0.92,-0.91) -- (7.58, 7.50) -- (-0.91,-0.90) -- (0.91, 2.10) -- (-20.73, -22.04) -- (22.68, 25.36) -- (15.48, 2.08) -- (14.09, 11.77) -- (7.57, 5.56) -- (-13.28, -10.35) -- (7.91, 6.17) -- (0.59, 2.34) -- (15.27, 15.86) -- (23.09, 21.73) -- (-6.52, -5.68) -- (-1.13, -0.85) -- (6.48, 6.21);
\end{scope}
\end{scope}
\begin{scope}[scale = 0.2,  xshift =35cm, yshift =-40cm]
	\draw[gray!80,->] (-1,0) -- (21,0) node at (23,0) {\footnotesize $x$};
	\draw[gray!80,->] (0,-1) -- (0,21) node at (0,23) {\footnotesize $y$};
	\draw[gray!80,] (20,-0.2) -- (20,0.2) node at (20.3,-2.2) {\footnotesize$6\%$};
	\draw[gray!80,] (-0.2,20) -- (0.2,20) node at (-2.4,20.2) {\footnotesize $6\%$};
	\node[OIorange] at (15,22) {\textbf{Cotton}};
	\begin{scope}[scale=0.33]
		\draw[OIorange] (-11.89, -12.01) -- (5.19, 18.65) -- (-20.69, -22.38) -- (-15.63, -25.99) -- (36.17, 34.95) -- (-5.33, -5.76) -- (3.38, 0.37) -- (-1.57, -5.66) -- (9.94, 11.08) -- (3.46, 1.46) -- (21.49, 17.91) -- (3.61, 6.64) -- (-9.99, 2.01) -- (5.84, 7.98) -- (33.54, 21.94) -- (15.89, 9.78) -- (-38.79, -42.07) -- (-18.34, -1.42) -- (4.45, 5.43);
\end{scope}
\end{scope}
\begin{scope}[scale = 0.2, xshift =35cm]
	\draw[gray!80,->] (-1,0) -- (21,0) node at (23,0) {\footnotesize $x$};
	\draw[gray!80,->] (0,-1) -- (0,21) node at (0,23) {\footnotesize $y$};
	\draw[gray!80,] (20,-0.2) -- (20,0.2) node at (20.3,-2.2) {\footnotesize$6\%$};
	\draw[gray!80,] (-0.2,20) -- (0.2,20) node at (-2.4,20.2) {\footnotesize $6\%$};
	\node[OIbluelight] at (15,22) {\small \textbf{Natural Gas}};
	\begin{scope}[scale=0.33]
	\draw[OIbluelight] (-19.57, -22.32) -- (-18.40, -16.93) -- (-10.46, -11.28) -- (32.60, 26.88) -- (22.70, 18.95) -- (-37.61, -36.58) -- (-34.18, -24.86) -- (-2.36, -1.58) -- (-10.75, -10.38) -- (44.50, 43.54) -- (-9.99, -12.77) -- (55.86, 53.46) -- (-16.22, -8.87) -- (-10.43, -7.07) -- (16.48, 21.48) -- (-24.39, -20.05) -- (-69.42, -37.76) -- (55.34, 57.39) -- (-1.52, -4.38);
\end{scope}
\end{scope}
	\begin{scope}[xshift=15cm]
	\node[black] at (5.5,6) {\textbf{October 2023}};
\begin{scope}[scale = 0.2]
	\draw[gray!80,->] (-1,0) -- (21,0) node at (23,0) {\footnotesize $x$};
	\draw[gray!80,->] (0,-1) -- (0,21) node at (0,23) {\footnotesize $y$};
	\draw[gray!80,] (20,-0.2) -- (20,0.2) node at (20.3,-2.2) {\footnotesize$6\%$};
	\draw[gray!80,] (-0.2,20) -- (0.2,20) node at (-2.4,20.2) {\footnotesize $6\%$};
	\node[copper] at (15,22) {\small \textbf{Copper}};
	\begin{scope}[scale=0.33]
		\draw[copper] (-6.35,-6.14) -- (-7.22,-7.10) -- (-7.40,-7.09) -- (18.50,18.44) -- (8.14,7.98) -- (-12.04,-11.42) -- (-0.13,0.13) -- (-4.39,-4.05) -- (-5.30,-5.28) -- (4.80,3.90) -- (2.77,1.13) -- (0.31,0.44) -- (2.89,2.82) -- (-5.84,-5.89) -- (2.39,2.64) -- (10.32,10.33) -- (-3.44,-2.90) -- (-6.45,-5.99) -- (14.49,14.32) -- (4.63,4.77) -- (-5.28,-4.67);
\end{scope}
\end{scope}
\begin{scope}[scale = 0.2, yshift =-40cm]
	\draw[gray!80,->] (-1,0) -- (21,0) node at (23,0) {\footnotesize $x$};
	\draw[gray!80,->] (0,-1) -- (0,21) node at (0,23) {\footnotesize $y$};
	\draw[gray!80,] (20,-0.2) -- (20,0.2) node at (20.3,-2.2) {\footnotesize$6\%$};
	\draw[gray!80,] (-0.2,20) -- (0.2,20) node at (-2.4,20.2) {\footnotesize $6\%$};
	\node[OIblack] at (15,22) {\small \textbf{Cocoa}};
	\begin{scope}[scale=0.33]
		\draw[OIblack] (-21.56,-20.03) -- (0.87,1.45) -- (-9.12,-7.90) -- (15.35,13.83) -- (-5.24,-4.63) -- (11.23,11.45) -- (-12.23,-10.41) -- (16.62,16.50) -- (2.57,2.55) -- (25.62,24.90) -- (12.10,12.56) -- (1.10,0.82) -- (2.47,3.00) -- (12.45,12.10) -- (17.29,17.69) -- (-8.05,-6.11) -- (1.61,2.12) -- (19.69,18.73) -- (10.90,9.28) -- (-6.80,-5.96) -- (-2.36,-2.60);
\end{scope}
\end{scope}
\begin{scope}[scale = 0.2, xshift =35cm, yshift =-40cm]
	\draw[gray!80,->] (-1,0) -- (21,0) node at (23,0) {\footnotesize $x$};
	\draw[gray!80,->] (0,-1) -- (0,21) node at (0,23) {\footnotesize $y$};
	\draw[gray!80,] (20,-0.2) -- (20,0.2) node at (20.3,-2.2) {\footnotesize$6\%$};
	\draw[gray!80,] (-0.2,20) -- (0.2,20) node at (-2.4,20.2) {\footnotesize $6\%$};
	\node[OIorange] at (15,22) {\small \textbf{Cotton}};
	\begin{scope}[scale=0.33]
		\draw[OIorange] (-3.67,-3.66) -- (-4.84,-4.83) -- (-5.45,-5.43) -- (6.91,6.89) -- (-2.08,-2.07) -- (-14.28,3.44) -- (-4.70,-6.58) -- (-1.53,-0.35) -- (13.25,12.65) -- (-10.69,-8.27) -- (-23.07,-18.96) -- (11.99,9.51) -- (0.36,-0.81) -- (-22.69,-19.52) -- (15.88,10.88) -- (-9.65,-9.21) -- (10.50,9.71) -- (9.22,7.08) -- (-2.49,0.12) -- (-17.73,-14.73) -- (-20.81,-16.41);
\end{scope}
\end{scope}
\begin{scope}[scale = 0.2, xshift =35cm]
	\draw[gray!80,->] (-1,0) -- (21,0) node at (23,0) {\footnotesize $x$};
	\draw[gray!80,->] (0,-1) -- (0,21) node at (0,23) {\footnotesize $y$};
	\draw[gray!80,] (20,-0.2) -- (20,0.2) node at (20.3,-2.2) {\footnotesize$6\%$};
	\draw[gray!80,] (-0.2,20) -- (0.2,20) node at (-2.4,20.2) {\footnotesize $6\%$};
	\node[OIbluelight] at (15,22) {\small \textbf{Natural Gas}};
	\draw[OIbluelight] (12.32,8.05) -- (1.46,0.70) -- (21.48,18.13) -- (17.18,11.90) -- (3.75,-0.92) -- (0.59,2.55) -- (-0.49,-0.09) -- (-3.29,-1.74) -- (-11.12,-5.12) -- (-13.62,-10.06) -- (-3.25,-1.06) -- (-2.51,-0.48) -- (-11.16,-13.73) -- (-6.67,-6.96) -- (3.08,1.53) -- (5.05,4.92) -- (4.32,5.33) -- (21.16,9.68) -- (-5.27,0.57) -- (18.70,9.84) -- (20.79,19.58);
\end{scope}
\end{scope}
\end{tikzpicture}
\caption{Paths of the first two future returns (front contract on $x$-axis, next contract on $y$) for different commodities in two representative months.}\label{fig:paths}
\end{figure}
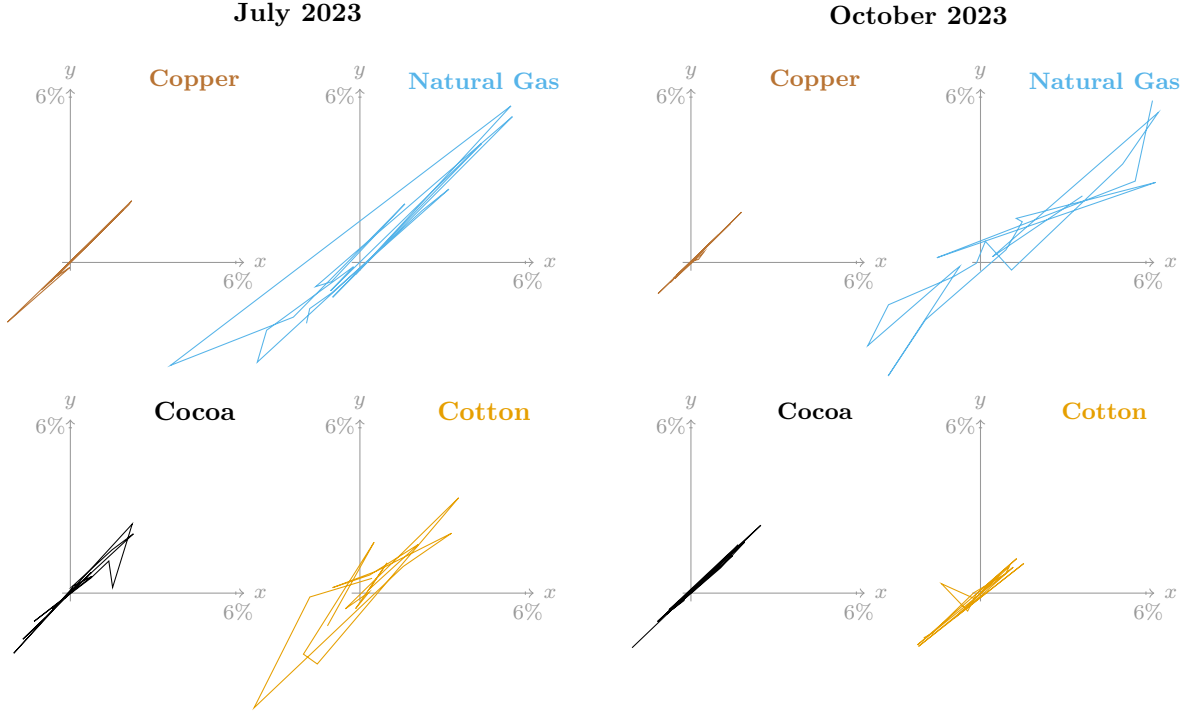

\section{Signature Approximations}\label{sec:approx}

We will now provide a rigorous argument for the expansion presented in Theorem \ref{thm:exp-first}.  To make things precise, we have to clarify the notion of distance with which we operate. As there is no obvious notion of intrinsic distance of signatures (notably, signatures do not form a vector space), we take the view of signatures as  sequences of random variables and will consider the distance induced from a norm on the ambient space. Specially, we endow the space of random sequences with a suitable family of norms, namely $L^p$-norms, properly weighted along the sequence.

\begin{definition}
We define the \emph{weighted $d$-signature $L^p$-norm} $\Vert \, \cdot \,  \Vert_{\sig(p,d)}$ on the space of sequences $\mathbb{Y} = (Y_0, Y_1, Y_2, \ldots)$ of random variables, $L^0(\Omega, \mathcal{F}, \mathbb{P}; \mathbb{R}^\mathbb{N})$, by setting
\[
\bigl\Vert \mathbb{Y} \bigr\Vert_{\sig(p,d,w)} := \bigl\Vert Y_0 \bigr\Vert_p + \sum_{m=1}^\infty \sum_{k=0}^{d^m-1} \frac{w_m}{m! d^m} \bigl\Vert Y_{d^{m-1} +k}\bigr\Vert_p.
\]
for weights $w \in \ell^1_{\geq 0}$, $\sum_{m=0}^\infty w_m = 1$ and $p \in [1,\infty)$, where $\Vert \, \cdot \, \Vert_p$ is the classical $L^p$-norm for ($\P$-null equivalence classes of) random variables.
\end{definition}
We note also that we have the Hardy-norms for continuous $(\mathbb{G}, \mathbb{P})$-semimartingales on $[0,1]$. Namely, let $Y$ be a continuous semimartingale, then it has a decomposition as $Y = Y_0 + M + A$, where $M$ is a local martingale and $A$ a process of finite total variation, $M_0 = A_0 =0$, and $Y_0$ is a constant. We then define
\[
\bigl\Vert Y \bigr\Vert_{\mathcal{H}^p} = \bigl\vert Y_0 \bigr\vert +\E\Bigl[ \bigl\langle M \bigr\rangle_1^\frac{p}{2} \Bigr]^\frac{1}{p} + \E\biggl[ \Bigl\vert \int_0^1 \bigl\vert dA_t \bigr\vert \Bigr\vert^p\biggr]^\frac{1}{p},
\] 
and denote the space of continuous $(\mathbb{G}, \mathbb{P})$-semimartingales with finite $\mathcal{H}^p$-norm by $\mathcal{H}^p$ while setting $\mathcal{H}^\infty := \bigcap_{p \geq 1} \mathcal{H}^p$.
\begin{proposition}\label{prop:GS-approx}
The approximations for $X^{(n)}$ for the spot returns and $C^{(n)}$ for the convenience yield satisfy for $\delta \leq 1$ the inequalities
\begin{equation}\label{eq:spot-approx}
\bigl\Vert X^\delta - X^{(n)} \bigr\Vert_{\mathcal{H}^p} \leq \bigl\Vert C^\delta - C^{(n)} \bigr\Vert_{\mathcal{H}^p} \leq \sqrt{\delta}^{n+1} \tilde{K}_1 \frac{\kappa^{n+1}}{\Gamma(\frac{n}{2}+1)} \leq \sqrt{\delta}^{n+1} K_1 
\end{equation} 
for some universal constants $K_1, \, \tilde{K}_1 >0$ independent of $p$ and $n$.
\end{proposition}

\begin{proof}
We start with the convenience yield. From the formal solution of the Ornstein--Uhlenbeck process we have
\[
C^\delta_t = c e^{-\delta \kappa t} + \theta(1-e^{- \delta\kappa  t}) + \sqrt{\delta} \gamma \int_0^t e^{-\delta \kappa(t-s)} \, \circ dW^2_s.
\]
and for the $n$-th order approximation.
\begin{align*}
C_t^{(n)} & = c + (c-\theta) \sum_{k=1}^{\lfloor \frac{n}{2}\rfloor} \delta^{k}\frac{(-\kappa)^{k}}{k!} t^{k} + \gamma \sqrt{\delta} \sum_{k=0}^{\lfloor \frac{n+1}{2}\rfloor-1} \delta^{k} \frac{(-\kappa)^{k}}{k!} \int_0^t (t-s)^{k} \, \circ dW^2_s.
\end{align*}
Thus, using the estimate $\sum_{k=0}^{2n+1}\frac{(-x)^k}{k!} \leq e^{-x} \leq \sum_{k=0}^{2n}\frac{(-x)^k}{k!}$, we conclude that
\begin{align*}
& \phantom{=::}  \bigl\Vert C^\delta - C^{(n)} \bigr\Vert_{\mathcal{H}^p}\\
& = \Biggl\Vert (c-\theta) \biggl(e^{- \delta\kappa  \cdot } - \sum_{k=0}^{\lfloor \frac{n}{2}\rfloor} \delta^{k}\frac{(-\kappa)^{k}}{k!} \cdot^{k} \biggr)  + \gamma \sqrt{\delta}  \int_0^\cdot  \biggl(e^{-\delta \kappa(\cdot -s)} -  \sum_{k=0}^{\lfloor \frac{n+1}{2}\rfloor -1} \delta^{k} \frac{(-\kappa)^{k}}{k!} (\cdot -s)^{k}\biggr) \, \circ dW^2_s \Biggr\Vert_{\mathcal{H}^p}\\
& = \bigl\vert c-\theta  \bigr \vert \cdot \Biggl\vert e^{- \delta\kappa } - \sum_{k=0}^{\lfloor \frac{n}{2}\rfloor} \delta^{k}\frac{(-\kappa)^{k}}{k!} \biggr\vert +\gamma \sqrt{\delta}  \E \Biggl[\Biggl( \int_0^1 \biggl(e^{-\delta \kappa(1-s)} -  \sum_{k=0}^{\lfloor \frac{n+1}{2}\rfloor -1} \delta^{k} \frac{(-\kappa)^{k}}{k!} (1-s)^{k}\biggr)^{2} \, ds\Biggr)^\frac{p}{2}\Biggr]^\frac{1}{p}\\
& \leq \bigl\vert c-\theta  \bigr\vert \cdot \delta^{(\lfloor \frac{n}{2}\rfloor+1)}\frac{\kappa^{(\lfloor \frac{n}{2}\rfloor+1)}}{\bigl((\lfloor \frac{n}{2}\rfloor+ 1)!)}  + \gamma \sqrt{\delta}  \Biggl( \int_0^1  \delta^{2\lfloor \frac{n+1}{2}\rfloor} \frac{(-\kappa)^{2\lfloor \frac{n+1}{2}\rfloor}}{(\lfloor \frac{n+1}{2}\rfloor !)^2} (1-s)^{2\lfloor \frac{n+1}{2}\rfloor} \, ds\Biggr)^\frac{1}{2}\\
& \leq \bigl\vert c-\theta  \bigr \vert \cdot \delta^{(\lfloor \frac{n}{2}\rfloor+1)}\frac{\kappa^{(\lfloor \frac{n}{2}\rfloor+1)}}{\bigl((\lfloor \frac{n}{2}\rfloor+ 1)!\bigr)}   + \gamma  \sqrt{\delta} \frac{\delta^{\lfloor \frac{n+1}{2}\rfloor } \kappa^{\lfloor \frac{n+1}{2}\rfloor}}{(\lfloor \frac{n+1}{2}\rfloor !) \cdot (2\lfloor \frac{n+1}{2}\rfloor +1)}.
\end{align*}
It follows for $\delta \leq 1$ that
\[
\bigl\Vert C^\delta - C^{(n)} \bigr\Vert_{\mathcal{H}^p} \leq \sqrt{\delta}^{n+1} \frac{\sqrt{\kappa}^{n+1}}{\Gamma(\frac{n}{2}+1)}\tilde{K}_1
\]
for the universal constant $\tilde{K}_1 := ( \kappa \vee 1)\bigl\vert c-\theta\bigr\vert + ( \kappa \wedge 1) \gamma$. To get the bound independent of the order of the approximation, we note that $\max_{z \in [0,\infty)} \frac{a^z}{\Gamma(z+1)} \leq \frac{a^{a+1}}{\Gamma(a+1)}$ for $a>0$ and thus $K_1 := \frac{\kappa^{\kappa+\frac{3}{2}}}{\Gamma(\kappa+1)}\tilde{K}_1$.

Next we consider the the spot returns $X$ that can be written explicitly as
\[
X^\delta_t = x_0 + \int_0^t \Bigl(r - \frac{\sigma^2}{2} - C_s^\delta \Bigr) \, ds + \sigma W^1_t.
\]
Combining this with equation \eqref{eq:Xn-rep} gives
\[
\bigl\Vert X^\delta - X^{(n)} \bigr\Vert_{\mathcal{H}^p}  = \biggl\Vert \int_0^\cdot C^\delta_s - C^{(n)}_s \,ds \biggr\Vert_{\mathcal{H}^p} 
\leq  \E\biggl[\sup_{t \in [0,1]}  \bigl\vert C^\delta_s - C^{(n)}_s \bigr\vert^p  \biggr]^\frac{1}{p}  \leq \bigl\Vert C^\delta - C^{(n)} \bigr\Vert_{\mathcal{H}^p}, \]
implying the result.
\end{proof}

We are now able to formulate and proof the main approximation result.
\begin{theorem}\label{thm:main}
The approximation $S_{0,1}\bigl(\mathbb{F}^{(n)}\bigr)$ of the signature of the term structure of futures log-prices $S_{0,1}\bigl(\mathbb{F}^{\delta}\bigr)$ satisfies for $\delta \leq 1$ that
\begin{equation}\label{eq:sig-approx}
\Bigl\Vert S_{0,1}\bigl(\mathbb{F}^{\delta}\bigr) - S_{0,1}\bigl(\mathbb{F}^{(n)}\bigr) \Bigr\Vert_{\sig(p,2,w)}  \leq \sqrt{\delta}^{n+1} \sqrt{p} K_2.
\end{equation}
for weights $w = (w_m)$, $w_m = \sqrt{2}^{m(1-m)}\bigl(\sum_{m=1}^\infty \sqrt{2}^{m(1-m)}\bigr)^{-1}$, and  some universal constant $K_2>0$ independent of $p$ and $n$.
\end{theorem}

Before proving this main result, we want to point out how the approximation \eqref{eq:main-approx} is a direct consequence of this theorem.

\begin{corollary}\label{cor:prac-approx}
There are constants $E_0, \,E_{1,0}, \,E_{1,1}, E_{2,1}, \, E_{2,2}, \, E_{3,1}, \, E_{3,2}, \,E_{3,3}, \,E_{3,4}$ and $E_{3,5}$, independent of $\kappa, \, \theta$ and $\gamma$, that for $\delta \leq 1$
\begin{align*}
&\biggl\Vert S_{0,1}\bigl(\mathbb{F}^{\delta}\bigr) -\biggl(E_0 +
 \Bigl(E_0 + \gamma E_{1,1}\Bigr) \sqrt{\delta} +\Bigl(E_0 + \gamma^k E_{2,1} + \kappa(\theta-c)E_{2,2}\Bigr)\delta  + \Bigl(E_0 + \gamma^k E_{3,1} + \bigl((\theta-c)\kappa\bigr)^k E_{3,2} \Bigr. \biggr.\\
&\phantom{==} \biggl. \Bigl. + \gamma \kappa E_{3,3} + \gamma^{k-1}(\theta-c)\kappa E_{3,4} +\gamma\bigl((\theta-c)\kappa\bigr)^{k-1}E_{3,5}\Bigr)\sqrt{\delta}^3 \biggr) \biggr\Vert_{\sig(p,2,w)}  \leq \delta^2 \sqrt{p} K_3.
\end{align*}
for weights $w = (w_m)$, $w_m = \sqrt{2}^{m(1-m)}\bigl(\sum_{m=1}^\infty \sqrt{2}^{m(1-m)}\bigr)^{-1}$, and some universal constant $K_3>0$.
\end{corollary}

\begin{proof}
See Appendix \ref{sec:proofs}.
\end{proof}

Now we can go to the proof of the main Theorem \ref{thm:main}. In a first step we will establish six estimates for the difference as well as the absolute size of integrals, that will be used in a next step for resolving iterated integrals. We relegate the proofs to Appendix \ref{sec:proofs}.

\begin{lemma}\label{lem:BIGlem}
For any $h \in \mathcal{H}^\infty$, any $n \in \mathbb{N}$, and $\delta \leq 1$ there exist constants $K_4, K_5, K_6, K_7, K_8 >0$ not depending on $p$ and $n$ such that
\begin{alignat*}{3}
& \text{o)} \quad \bigl\Vert h_t \bigr\Vert_p \leq  2 \bigl(1+\sqrt{2}\bigr) \sqrt{p}\bigl\Vert  h \bigr\Vert_{\mathcal{H}^{p}},\\
& \text{i)} \quad \biggl\Vert \int_0^\cdot h_s \, \circ dC^\delta_s \biggr\Vert_{\mathcal{H}^p} \leq   \sqrt{\delta} K_4 p \bigl\Vert  h \bigr\Vert_{\mathcal{H}^{2p}}, & \quad 
& \text{ii)} \quad \biggl\Vert \int_0^\cdot h_s \, \circ dX^\delta_s \biggr\Vert_{\mathcal{H}^p} \leq K_5 p \bigl\Vert  h \bigr\Vert_{\mathcal{H}^{2p}},\\
& \text{iii)} \quad \biggl\Vert \int_0^\cdot h_s \, \circ dC^\delta_s - \int_0^t h_s \, \circ dC^{(n)}_s\biggr\Vert_{\mathcal{H}^p} \leq  \sqrt{\delta}^{n+1} K_6 p \bigl\Vert  h \bigr\Vert_{\mathcal{H}^{2p}},\\
& \text{iv)} \quad \biggl\Vert \int_0^\cdot h_s \, \circ dX^\delta_s - \int_0^t h_s \, \circ dX^{(n)}_s\biggr\Vert_{\mathcal{H}^p} \leq  \sqrt{\delta}^{n+1}K_6  p \bigl\Vert  h \bigr\Vert_{\mathcal{H}^{2p}},\\
& \text{v)} \quad \biggl\Vert \int_0^\cdot h_s \,\circ  dC^{(n)}_s \biggr\Vert_{\mathcal{H}^p} \leq   \sqrt{\delta} K_7 p \bigl\Vert  h \bigr\Vert_{\mathcal{H}^{2p}}, & \quad
&\text{vi)} \quad \biggl\Vert \int_0^\cdot h_s \, \circ dX^{(n)}_s \biggr\Vert_H \leq K_8 p \bigl\Vert  h \bigr\Vert_{\mathcal{H}^{2p}}.
\end{alignat*}
\end{lemma}

\begin{proof}
See Appendix \ref{sec:proofs}.
\end{proof}

Based on the lemma above, we are able to provide an estimate for the approximation of the iterated integrals of the two-dimensional process $(X,C)$.

\begin{lemma}\label{lem:it-int-approx}
For the approximation of the iterated integrals it holds for $i_1, \ldots, i_k \in \{1,2\}$ that
\[
 \Bigl\Vert S^{i_1, \ldots, i_k}_{0,1}\bigl((X^\delta, C^\delta)\bigr) - S^{i_1, \ldots, i_k}_{0,1}\bigl((X^{(n)}, C^{(n)})\bigr) \Bigr\Vert_p  \leq \sqrt{\delta}^{n+1}  K_9 p^{k} \sqrt{p} \sqrt{2}^{k(k-1)}
\]
for some universal constants $K_9$ (in particular independent of $p$, $k$ and $n$).
\end{lemma}
\begin{proof}
For arbitrary choices $Z^j \in \{X, C\}$, $j \in \{1, \cdots, k\}$, with slow process $Z^{j,\delta}$ and approximation $Z^{j,(n)}$ we have
\begin{align*}
& \phantom{=::}\biggl\Vert \int_0^1 \int_0^{t_k} \cdots \int_0^{t_2} 1 \, \circ dZ^{1,\delta}_{t_1}  \cdots \circ dZ^{k-1,\delta}_{t_{k-1}} \circ dZ^{k,\delta}_{t_k} - \int_0^1 \int_0^{t_k} \cdots \int_0^{t_2} 1 \, \circ dZ^{1,(n)}_{t_1} \cdots \circ dZ^{k-1,(n)}_{t_{k-1}} \circ dZ^{k,(n)}_{t_k} \biggr\Vert_p \\
&= \biggl\Vert \sum_{l=1}^k \int_0^1  \cdots \int_0^{t_{l+1}} \cdots \int_0^{t_2} 1 \, \circ dZ^{1,\delta}_{t_1}\cdots \circ dZ^{{l-1},\delta}_{t_{l-1}} \bigl(\circ dZ^{l,\delta}_{t_l} -  \circ dZ^{(n)}_{t_l}\bigr) \circ dZ^{l+1,(n)}_{t_{l+1}}\cdots  \circ dZ^{k,(n)}_{t_k}\biggr\Vert_p \\
& \leq  \sum_{l=1}^k\biggl\Vert  \int_0^1  \cdots \int_0^{t_{l+1}} \cdots \int_0^{t_2} 1 \, \circ dZ^{1,\delta}_{t_1} \cdots \circ dZ^{l-1,\delta}_{t_{l-1}} \bigl(\circ dZ^{l,\delta}_{t_l} -  \circ dZ^{l,(n)}_{t_l}\bigr) \circ dZ^{l+1,(n)}_{t_{l+1}} \cdots \circ dZ^{k,(n)}_{t_k}\biggr\Vert_p \\
& \leq \sum_{l=1}^k\biggl\Vert \int_0^\cdot \cdots\int_0^{t_2} 1 \, \circ dZ^{1,\delta}_{t_1} \cdots \circ dZ^{l-1,\delta}_{t_{l-1}} \bigl(\circ dZ^{l,\delta}_{t_l} -  \circ dZ^{l,(n)}_{t_l}\bigr) \biggr\Vert_{\mathcal{H}^{2^{k-l}p}}  2 \bigl(1+\sqrt{2}\bigr) \sqrt{p} (K_7 \vee K_8)^{k-l}\prod_{j=0}^{k-l-1} 2^jp\\
& \leq \sum_{l=1}^k\biggl\Vert \int_0^\cdot \cdots \int_0^{t_2} 1 \, \circ dZ^{1,\delta}_{t_1} \cdots \circ dZ^{l-1,\delta}_{t_{l-1}} \biggr\Vert_{\mathcal{H}^{2^{k-l+1}p}}  \sqrt{\delta}^{n+1}  2 \bigl(1+\sqrt{2}\bigr) \sqrt{p}K_6 (K_7 \vee K_8)^{k-l}\prod_{j=0}^{k-l} 2^jp\\
& \leq \sum_{l=1}^k\bigl\Vert 1 \bigr\Vert_{\mathcal{H}^{2^kp}}  (K_4 \vee K_5)^{l-1} \sqrt{\delta}^{n+1}  2 \bigl(1+\sqrt{2}\bigr) \sqrt{p} K_6  (K_7 \vee K_8)^{k-l}\prod_{j=0}^{k-1} 2^jp \leq \sqrt{\delta}^{n+1}K_9 \sqrt{2}^{k(k-1)}p^{k}\sqrt{p} 
\end{align*}
by applying Lemma \ref{lem:BIGlem} repeatedly. The constant is chosen as $K_9 = 2 \bigl(1+\sqrt{2}\bigr)\bigl(K_4 \vee K_5 \vee K_6 \vee K_7 \vee K_8\bigr)$.
\end{proof}

Finally we can assemble the proof of the main theorem.
\begin{proof}[Proof of Theorem \ref{thm:main}]
We note first that for the approximation \eqref{eq:GS-approx} of the $B$-term \eqref{eq:GS-exp} it holds for every $k \in \{1,2,\ldots,n\}$ and $n \in \N$  that
\[
\Bigl\vert B^{\delta}(T_k) - B^{(n)}(T_k) \Bigr\vert \leq  \frac{(\delta \kappa)^{n} (T_k)^{n+1}}{(n+1)!}.
\]
Using now the approximation \eqref{eq:GS-approx} we have for $\delta \leq 1$ for the iterated term structure integrals
\begin{align}
& \phantom{==} \Bigl\Vert S^{i_1, \ldots, i_k}_{0,1}\bigl(\mathbb{F}^{\delta}\bigr) - S^{i_1, \ldots, i_k}_{0,1}\bigl(\mathbb{F}^{(n)}\bigr) \Bigr\Vert_p  \nonumber \\
& \leq \sum_{(j_1, \ldots, j_k) \in \{1,2\}^k} \biggl\Vert \prod_{l=1}^k B^{\delta}(T_{i_l})^{j_l-1} S^{j_1, \ldots, j_k}_{0,1}\bigl((X^{\delta},C^{\delta})\bigr) - \prod_{l=1}^k B^{(\lfloor \frac{n}{2}\rfloor)}(T_{i_l})^{j_l-1} S^{j_1, \ldots, j_k}_{0,1}\bigl((X^{(n)},C^{(n)})\bigr) \biggr\Vert_p \nonumber \\
& \leq \sum_{(j_1, \ldots, j_k) \in \{1,2\}^k} \biggl( \prod_{l=1}^k \Bigl\vert B^{\delta}(T_{i_l})^{j_l-1} - B^{(\lfloor \frac{n}{2}\rfloor)}(T_{i_l})^{j_l-1} \Bigr\vert \cdot \Bigl\Vert  S^{j_1, \ldots, j_k}_{0,1}\bigl((X^{\delta},C^{\delta})\bigr)  \Bigr\Vert_p \biggr. \nonumber \\
& \phantom{==} \biggl. + \prod_{l=1}^k B^{(\lfloor \frac{n}{2}\rfloor)}(T_{i_l})^{j_l-1} \biggl\Vert S^{j_1, \ldots, j_k}_{0,1}\bigl((X^{\delta},C^{\delta})\bigr) -  S^{j_1, \ldots, j_k}_{0,1}\bigl((X^{(n)},C^{(n)})\bigr) \biggr\Vert_p \biggr) \label{eq:mid-proof}\\
& \leq \Bigl(\frac{(\delta \kappa)^{\lfloor \frac{n}{2}\rfloor +1} (T_d)^{\lfloor \frac{n}{2}\rfloor +2}}{(\lfloor \frac{n}{2}\rfloor+2)!}\Bigr)^{k}\sqrt{\delta} \bigl(K_4 \vee K_5\bigr)^k\sqrt{2}^{k(k-1)}p^k+\Bigl(\frac{(\delta \kappa)^{\lfloor \frac{n}{2}\rfloor} (T_d)^{\lfloor \frac{n}{2}\rfloor +1}}{(\lfloor \frac{n}{2}\rfloor+1)!}\Bigr)^{k} \sqrt{\delta}^{n+1} K_9 \sqrt{2}^{k(k-1)}p^{k}\sqrt{p} \nonumber  \\
&
\leq \sqrt{\delta}^{n+1} \bigl(K_{10}\bigr)^k \sqrt{2}^{k(k-1)}p^{k}\sqrt{p}\nonumber 
\end{align}
by Lemma \ref{lem:BIGlem} for the first term and Lemma \ref{lem:it-int-approx} for some constant 
\[
K_{10} := \frac{1}{\kappa} \frac{(\kappa T_d)^{\kappa T_d +1}}{\Gamma(\kappa T_d +1)}\bigl(K_4 \vee K_5 + K_9\bigr).
\]
Note that for the first term in \eqref{eq:mid-proof} not being zero, one of the iterated integrals has to be with respect to $C^\delta$. Thus we can conclude
\begin{align*}
 \Bigl\Vert S_{0,1}\bigl(\mathbb{F}^{\delta}\bigr) - S_{0,1}\bigl(\mathbb{F}^{(n)}\bigr) \Bigr\Vert_{\sig(p,2,w)}  &= \sum_{m=1}^\infty \sum_{(i_1, \ldots, i_m) \in \{1,2\}^m} \frac{w_m}{m! (2p)^m } \Bigl\Vert S^{i_1, \ldots, i_m}_{0,1}\bigl(\mathbb{F}^{\delta}\bigr) - S^{i_1, \ldots, i_m}_{0,1}\bigl(\mathbb{F}^{(n)}\bigr) \Bigr\Vert_p \\
& \leq  \sum_{m=1}^\infty \sum_{(i_1, \ldots, i_k) \in \{1,2\}^k} \frac{w_m}{m! (2p)^m } \sqrt{\delta}^{n+1}  \bigl(K_{10}p\bigr)^m \sqrt{2}^{m(m-1)}\sqrt{p}\\
& \leq \sqrt{\delta}^{n+1}  \sum_{m=1}^\infty \frac{(K_{10})^m}{m!} \sqrt{p} \leq \sqrt{\delta}^{n+1} e^{K_{10}} \sqrt{p}.
\end{align*}
Thus the claim holds for $K_2 := e^{K_{10}}$.
\end{proof}

\section{Conclusion}\label{sec:conclusion}

We have shown that changes in the signature vector of commodity term structure returns can be interpreted in terms of model parameters along a perturbative expansion. In particular, we established a concrete relationship between the signature-based feature set and model parameters of the Gibson--Schwartz model. In this framework, convenience yield volatility is the parameter with largest impact on signature variations. We have effectively developed a theory of signature approximation for the term structure of commodities, in an appropriate normed space. As there is no obvious notion of distance for signatures, we have embedded them in the space of sequences of random variables and assigned a weighted sequence norm on the space.

More practically, we have shown that the combination of a purely data-driven signature method and more classical, model-based approaches can lead to an interpretable data set. We suggest that perturbation methods are particularly useful when interpreting signatures relative to the parameters in a model whose dynamics are well-understood by practitioners. Observe that we have restricted ourselves to a regular perturbation approach for the commodity futures curve, which seems reasonable in this specific case. For other data sets, we surmise that a true multi-scale perturbation may be more suitable. We are hopeful that this specific case study opens a door for further research on signature approximations, in an effort to produce interpretable feature sets.

\section*{Acknowledgments} This article builds on empirical work  for the Research Experience for Undergraduates in Industrial Mathematics and Statistics at Worcester Polytechnic Institute (WPI), funded by the National Science Foundation Award DMS-2244306, see \cite{IKMSS25}. The authors want to thank the students involved in this project, Tora Ito (WPI), Adam Mullaney (WPI) and Kathleen Shiffer (Swarthmore College) for their great work that inspired the current this article.

\section*{Disclosure Statement} No potential conﬂict of interest was reported by the author(s).

\section*{Data Availability Statement} The data for Figure \ref{fig:paths} were retrieved from Allasso SA, \url{https://www.allasso.ch/}.

\bibliographystyle{alpha}
\bibliography{sig-bib}

\appendix

\section{Proofs}\label{sec:proofs}

In this appendix we assemble the more technical proofs of the results of the paper and present them in order of appearance.

\begin{proof}[Proof of Corollary \ref{cor:prac-approx}] 
We note that for
\[
S^{i_1, \ldots, i_k}_{0,1}\bigl(\mathbb{F}^{(3)}\bigr) = \sum_{(j_1, \ldots, j_k) \in \{1,2\}^k} \prod_{l=1}^k B^{(1)}(T_{i_l})^{j_l-1} S^{j_1, \ldots, j_k}_{0,1}\bigl((X^{(3)},C^{(3)})\bigr)
\]
it follows from
\[
B^{(1)}(T_{i_l}) = T_{i_l} - \delta \frac{\kappa}{2} \bigl(T_{i_l}\bigr)^2
\]
that
\begin{align*}
S^{i_1, \ldots, i_k}_{0,1}\bigl(\mathbb{F}^{(3)}\bigr) & = S^{1, \ldots, 1}_{0,1}\bigl((X^{(3)},C^{(3)}) + \sum_{\substack{(j_1, \ldots, j_k) \in \{1,2\}^k\\ j_1 + \cdots j_k = k+1}} T_{i_l}\ind_{\{j_l=2\}} S^{j_1, \ldots, j_k}_{0,1}\bigl((X^{(3)},C^{(3)})\bigr) \\& \phantom{:=} - \delta \frac{\kappa}{2}   \sum_{\substack{(j_1, \ldots, j_k) \in \{1,2\}^k\\ j_1 + \cdots j_k = k+1}} \bigl(T_{i_l}\bigr)^2 \ind_{\{j_l=2\}} S^{j_1, \ldots, j_k}_{0,1}\bigl((X^{(1)},C^{(1)})\bigr) + O(\delta^2).
\end{align*}
Moreover, as
\begin{align*}
S^{j_1, \ldots, j_k}_{0,1}\bigl((X^{(3)},C^{(3)})\bigr) & = \sqrt{\delta} S^{j_1, \ldots, j_k}_{0,1}\bigl((\hat{X}^{(0)},\hat{C}^{(1)})\bigr) + \delta\Bigl( S^{j_1, \ldots, j_k}_{0,1}\bigl((\hat{X}^{(0)},\hat{C}^{(2)})\bigr)  + S^{j_1, \ldots, j_k}_{0,1}\bigl((\hat{X}^{(1)},\hat{C}^{(1)})\bigr)\Bigr) \\
&\phantom{=:} + \sqrt{\delta}^3\Bigl( S^{j_1, \ldots, j_k}_{0,1}\bigl((\hat{X}^{(0)},\hat{C}^{(3)})\bigr)  + S^{j_1, \ldots, j_k}_{0,1}\bigl((\hat{X}^{(1)},\hat{C}^{(2)})\bigr) + S^{j_1, \ldots, j_k}_{0,1}\bigl((\hat{X}^{(2)},\hat{C}^{(1)})\bigr) \Bigr) \\ &\phantom{:=}+ O\bigl(\delta^{2}\bigr)
\end{align*}
(as the terms with non-zero $C^{(0)}$ terms vanish as they include an integration against a constant), we have
\begin{align*}
S^{i_1, \ldots, i_k}_{0,1}\bigl(\mathbb{F}^{(3)}\bigr)  & = S^{1, \ldots, 1}_{0,1}\bigl((\hat{X}^{(0)},\hat{C}^{(0)})\bigr) + \sqrt{\delta} S^{1, \ldots, 1}_{0,1}\bigl((\hat{X}^{(0)},\hat{C}^{(1)})\bigr)  \\
&\phantom{==}+ \delta\Bigl( S^{1, \ldots, 1}_{0,1}\bigl((\hat{X}^{(0)},\hat{C}^{(2)})\bigr)  + S^{1, \ldots, 1}_{0,1}\bigl((\hat{X}^{(1)},\hat{C}^{(1)})\bigr)\Bigr) \\
&\phantom{==} + \sqrt{\delta}^3\Bigl( S^{1, \ldots, 1}_{0,1}\bigl((\hat{X}^{(0)},\hat{C}^{(3)})\bigr)  + S^{1, \ldots, 1}_{0,1}\bigl((\hat{X}^{(1)},\hat{C}^{(2)})\bigr) +S^{1, \ldots, 1}_{0,1} \bigl((\hat{X}^{(2)},\hat{C}^{(1)})\bigr) \Bigr. \\
&\phantom{==} + \Bigl. \sum_{\substack{(j_1, \ldots, j_k) \in \{1,2\}^k\\ j_1 + \cdots j_k = k+1}} T_{i_l}\ind_{\{j_l=2\}} \biggl(\sqrt{\delta} S^{j_1, \ldots, j_k}_{0,1}\bigl((\hat{X}^{(0)},\hat{C}^{(1)})\bigr) \biggr. \\
&\phantom{==} \biggl.+ \delta\Bigl( S^{j_1, \ldots, j_k}_{0,1}\bigl((\hat{X}^{(0)},\hat{C}^{(2)})\bigr)  + S^{j_1, \ldots, j_k}_{0,1}\bigl((\hat{X}^{(1)},\hat{C}^{(1)})\bigr)\Bigr) \\
&\phantom{==} + \sqrt{\delta}^3\Bigl( S^{j_1, \ldots, j_k}_{0,1}\bigl((\hat{X}^{(0)},\hat{C}^{(3)})\bigr)  + S^{j_1, \ldots, j_k}_{0,1}\bigl((\hat{X}^{(1)},\hat{C}^{(2)})\bigr) + S^{j_1, \ldots, j_k}_{0,1}\bigl((\hat{X}^{(2)},\hat{C}^{(1)})\bigr) \Bigr)\biggr)\\
&\phantom{==} - \sqrt{\delta}^3 \frac{\kappa}{2}   \sum_{\substack{(j_1, \ldots, j_k) \in \{1,2\}^k\\ j_1 + \cdots j_k = k+1}} \bigl(T_{i_l}\bigr)^2 \ind_{\{j_l=2\}} S^{j_1, \ldots, j_k}_{0,1}\bigl((\hat{X}^{(0)},\hat{C}^{(1)})\bigr)
 + O\bigl(\delta^{2}\bigr).
\end{align*}
Calculating the terms explicitly, we have (with a slight abuse of notation, writing $Y_t$ instead of $Y = (Y_t)_{t \in [0,1]})$ for stochastic processes $Y$),
\begin{align*}
S^{1, \ldots, 1}_{0,1}\bigl((\hat{X}^{(0)},\hat{C}^{(0)})\bigr) &  = S^{1, \ldots, 1}_{0,1}\bigl((\hat{X}^{(0)},\hat{C}^{(1)})\bigr) = S^{1, \ldots, 1}_{0,1}\bigl((\hat{X}^{(0)},\hat{C}^{(2)})\bigr) = S^{1, \ldots, 1}_{0,1}\bigl((\hat{X}^{(0)},\hat{C}^{(3)})\bigr) \\ & = S^{1, \ldots, 1}_{0,1}\biggl(- \bigl(c- r+\frac{\sigma^2}{2}\bigr) t + \sigma W_t^1\biggr)\\
S^{1, \ldots, 1}_{0,1}\bigl((\hat{X}^{(1)},\hat{C}^{(1)})\bigr) & = S^{1, \ldots, 1}_{0,1}\bigl((\hat{X}^{(1)},\hat{C}^{(2)})\bigr)  = \gamma^k S^{1, \ldots, 1}_{0,1}\biggl(\int_0^t (t -s) \, \circ dW^2_s\biggr)  \\
S^{1, \ldots, 1}_{0,1}\bigl((\hat{X}^{(2)},\hat{C}^{(1)})\bigr) & = (\theta-c)^k\frac{\kappa^k}{2^k} S^{1, \ldots, 1}_{0,1}\bigl(t^2\bigr) = (\theta-c)^k\frac{\kappa^k}{2^k \cdot k!} t^{2k}  \\
\ind_{\{j_1 + \cdots j_k = k+1\}} S^{j_1, \ldots, j_k}_{0,1}\bigl((\hat{X}^{(0)},\hat{C}^{(1)})\bigr) & = \gamma \ind_{\{j_1 + \cdots j_k = k+1\}} S^{j_1, \ldots, j_k}_{0,1}\biggl(\Bigl(- \bigl(c- r+\frac{\sigma^2}{2}\bigr) t + \sigma W_t^1, \circ W_t^2\Bigr)\biggr)\\
\ind_{\{j_1 + \cdots j_k = k+1\}} S^{j_1, \ldots, j_k}_{0,1}\bigl((\hat{X}^{(0)},\hat{C}^{(2)})\bigr) & =(\theta-c)\kappa \ind_{\{j_1 + \cdots j_k = k+1\}} S^{j_1, \ldots, j_k}_{0,1}\biggl(\Bigl(- \bigl(c- r+\frac{\sigma^2}{2}\bigr) t + \sigma W_t^1,t\Bigr)\biggr)\\
\ind_{\{j_1 + \cdots j_k = k+1\}} S^{j_1, \ldots, j_k}_{0,1}\bigl((\hat{X}^{(0)},\hat{C}^{(3)})\bigr) & = -\kappa \gamma \ind_{\{j_1 + \cdots j_k = k+1\}} \\
& \phantom{==}S^{j_1, \ldots, j_k}_{0,1}\biggl(\Bigl(- \bigl(c- r+\frac{\sigma^2}{2}\bigr) t + \sigma W_t^1, \int_0^t ( t -s) \, \circ dW^2_s\Bigr)\biggr)\\
\ind_{\{j_1 + \cdots j_k = k+1\}} S^{j_1, \ldots, j_k}_{0,1}\bigl((\hat{X}^{(1)},\hat{C}^{(1)})\bigr) & =\gamma^k \ind_{\{j_1 + \cdots j_k = k+1\}} S^{j_1, \ldots, j_k}_{0,1}\biggl(\Bigl(\int_0^t (t -s) \, \circ dW^2_s, W_t^2\Bigr)_{t \in [0,1]}\biggr)\\
\ind_{\{j_1 + \cdots j_k = k+1\}} S^{j_1, \ldots, j_k}_{0,1}\bigl((\hat{X}^{(1)},\hat{C}^{(2)})\bigr) & =\gamma^{k-1}(\theta-c)\kappa \ind_{\{j_1 + \cdots j_k = k+1\}} S^{j_1, \ldots, j_k}_{0,1}\biggl(\Bigl(\int_0^t (t -s) \, \circ dW^2_s,t\Bigr)\biggr)\\
\ind_{\{j_1 + \cdots j_k = k+1\}} S^{j_1, \ldots, j_k}_{0,1}\bigl((\hat{X}^{(2)},\hat{C}^{(1)})\bigr) & =\gamma (\theta-c)^{k-1}\frac{\kappa^{k-1}}{2^{k-1}} \ind_{\{j_1 + \cdots j_k = k+1\}} S^{j_1, \ldots, j_k}_{0,1}\bigl(\bigl(t , W_t^2\bigr)\Bigr)
\end{align*}
Combining terms with the same pre-factors yields the result. The constant $K_3$ can be calculated explicitly by writing out the (finitely many) terms of order $\delta^2$ and higher (up to order $\delta^{k+1}$) and estimating them.
\end{proof}

\begin{proof}[Proof of Lemma \ref{lem:BIGlem}]
We start by deriving some general estimates for further use. By the martingale representation theorem in the Brownian filtration, we know that $h \in \mathcal{H}^\infty$ has the decomposition
\[
h_t = h_0 + A_t + M_t = h_0 + A_t + \int_0^t \beta_t^1 \, dW_t^1 + \int_0^t \beta_t^2 \, dW_t^2.
\]
where $A$ is a process of finite total variation, $M$ a local martingale and $\beta^1$, $\beta^2$ are square integrable, predictable processes.
If follows in particular that for $i \in \{1,2\}$ we have
\begin{equation}\label{eq:Strat-comp}
\biggl\Vert \int_0^\cdot \beta^i_s \, ds \biggr\Vert_{\mathcal{H}^p} \leq \E \biggl[ \Bigl(\int_0^1 \bigl\vert \beta^i_s\bigr\vert \, ds\Bigr)^p \biggr]^\frac{1}{p} \leq \E \biggl[ \Bigl\langle\int_0^1  \beta^i_s \, dW^i \Bigr\rangle^\frac{p}{2} \biggr]^\frac{1}{p} \leq \bigl\Vert h \bigr\Vert_{\mathcal{H}^p}.
\end{equation}
Moreover, we have
\begin{align*}
\bigl\Vert h_t\bigr\Vert_p  &  \leq  \bigl \vert h_0\bigr\vert + \E\biggl[\sup_{t \in [0,1]} \vert A_t\vert^p\Bigr]^\frac{1}{p}  +  \E\biggl[\sup_{t \in [0,1]} \vert M_t \vert^p \Bigr]^\frac{1}{p} 
\leq \bigl\Vert h_0\bigr\Vert_{\mathcal{H}^p} + \bigl\Vert A \bigr\Vert_{\mathcal{H}^p} +  2\sqrt{2p} \E\Bigl[ \bigl\langle M \bigr\rangle_1^\frac{p}{2}\Bigr]^\frac{1}{p} \\& \leq 2\bigl(1+\sqrt{2}\bigr)  \sqrt{p} \bigl\Vert h \bigr\Vert_{\mathcal{H}^p},
\end{align*}
by the Burkholder--Davis--Gundy (BDG) inequality for continuous local martingales (\cite[Theorem 2]{Ren08}), proving o). More generally, for any bounded deterministic function $b$ with $\overline{b} := \sup_{t \in [0,1]} \vert b_t \vert $ we have again by the BDG inequality that
\begin{align}
\biggl\Vert \int_0^\cdot h_s b_s \, dW^i_s \biggr\Vert_{\mathcal{H}^p} & \leq \biggl\Vert \int_0^\cdot h_0 b_s \, dW^i_s \biggr\Vert_{\mathcal{H}^p} + \biggl\Vert \int_0^\cdot A_s b_s \, dW^i_s \biggr\Vert_{\mathcal{H}^p} + \biggl\Vert \int_0^\cdot M_s b_s \, dW^i_s \biggr\Vert_{\mathcal{H}^p}\nonumber \\
& \leq \overline{b} \E\biggl[ \Bigl(\int_0^1 h_0^2 \, dt \Bigr)^\frac{p}{2}\biggr]^\frac{1}{p} + \overline{b} \E\biggl[ \Bigl(\int_0^1 A_t^2 \, dt \Bigr)^\frac{p}{2}\biggr]^\frac{1}{p} + \overline{b} \E\biggl[ \Bigl(\int_0^1 M_t^2 \, dt \Bigr)^\frac{p}{2}\biggr]^\frac{1}{p} \nonumber\\
&  \leq \overline{b} \bigl \vert h_0\bigr\vert + \overline{b} \E\biggl[\sup_{t \in [0,1]} \vert A_t\vert^p\Bigr]^\frac{1}{p}  + \overline{b} \E\biggl[\sup_{t \in [0,1]} \vert M_t \vert^p \Bigr]^\frac{1}{p} \nonumber \\
& \leq \overline{b} \bigl\Vert h_0\bigr\Vert_{\mathcal{H}^p} + \overline{b}\bigl\Vert A \bigr\Vert_{\mathcal{H}^p} +  2\sqrt{2p} \overline{b} \E\Bigl[ \bigl\langle M \bigr\rangle_1^\frac{p}{2}\Bigr]^\frac{1}{p}  \leq 2\bigl(1+\sqrt{2}\bigr)  \sqrt{p} \overline{b}\bigl\Vert h \bigr\Vert_{\mathcal{H}^p}.\label{eq:SI-est}
\end{align}
Additionally, using the Cauchy--Schwarz inequality together with the BDG inequality gives
\begin{align}
& \phantom{=:}\biggl\Vert \int_0^\cdot h_s \int_0^s b_u \, dW^2_u ds \biggr\Vert_{\mathcal{H}^p} \nonumber \\
& \leq \bigl\vert h_0\bigr\vert \E\biggl[ \Bigl( \int_0^1 \Bigl\vert  \int_0^t b_s \, dW^2_s \Bigr\vert dt \Bigr)^p \biggr]^\frac{1}{p}  + \E\biggl[ \Bigl( \int_0^1 \Bigl\vert A_t \int_0^t b_s \, dW^2_s \Bigr\vert dt \Bigr)^p \biggr]^\frac{1}{p} + \E\biggl[ \sup_{t \in [0,1]} \Bigl\vert M_s \int_0^s b_u \, dW^2_u \Bigr\vert^p \biggr]^\frac{1}{p} \nonumber \\
&\leq  \bigl\vert h_0\bigr\vert \E\biggl[ \sup_{t \in [0,1]}  \Bigl\vert \int_0^t b_s \, dW^2_s \Bigr\vert^p\biggr]^\frac{1}{p} +\E\biggl[ \Bigl\vert \int_0^1 \bigl\vert dA_t \bigr\vert \Bigr\vert^p \cdot \sup_{t \in [0,1]}  \Bigl\vert \int_0^t b_s \, dW^2_s \Bigr\vert^p\biggr]^\frac{1}{p} \nonumber \\
& \phantom{==} +
\E\biggl[ \sup_{t \in [0,1]} \bigl\vert M_s \bigr\vert^p \cdot \sup_{t \in [0,1]}  \Bigl\vert \int_0^t b_s \, dW^2_s \Bigr\vert^p \biggr]^\frac{1}{p} \nonumber \\
&\leq 4\sqrt{p}  \bigl\vert h_0\bigr\vert  + \E\biggl[ \Bigl\vert \int_0^1 \bigl\vert dA_t \bigr\vert \Bigr\vert^{2p}  \biggr]^\frac{1}{2p} \cdot \E\biggl[ \sup_{t \in [0,1]} \Bigl\vert \int_0^s b_u \, dW^2_u \Bigr\vert^{2p} \biggr]^\frac{1}{2p} \nonumber\\
&\phantom{==}
+ \E\biggl[ \sup_{t \in [0,1]} \bigl\vert M_s \bigr\vert^{2p}  \biggr]^\frac{1}{2p} \cdot \E\biggl[ \sup_{t \in [0,1]} \Bigl\vert \int_0^s b_u \, dW^2_u \Bigr\vert^{2p} \biggr]^\frac{1}{2p} \nonumber\\
&\leq 2\sqrt{2}\sqrt{p} \overline{b} \bigl\vert h_0\bigr\vert  +  2\sqrt{2}\sqrt{p} \overline{b} \E\biggl[ \Bigl\vert \int_0^1 \bigl\vert dA_t \bigr\vert \Bigr\vert^{2p}  \biggr]^\frac{1}{2p}  + 2\sqrt{2}\sqrt{p} \E\Bigl[ \bigl\langle M \bigr\rangle^{p}  \Bigr]^\frac{1}{2p} \cdot 2\sqrt{2}\sqrt{p}\E\biggl[ \Bigl( \int_0^1 b_s^2 \, ds \Bigr)^p \biggr]^\frac{1}{2p} \nonumber \\
&\leq 2\sqrt{2}\sqrt{p} \overline{b} \bigl\Vert h \bigr\Vert_{\mathcal{H}^{2p}} + 2\sqrt{2}\sqrt{p} \overline{b} \bigl\Vert h \bigr\Vert_{\mathcal{H}^{2p}} + 8p \overline{b}\bigl\Vert h \bigr\Vert_{\mathcal{H}^{2p}} \leq 4\bigl(2 + \sqrt{2}\bigr)p \overline{b}\bigl\Vert h \bigr\Vert_{\mathcal{H}^{2p}} \label{eq:prod-ineq}
\end{align}
as $p \geq 1$. Based on these preliminary results we can now prove the estimates of the lemma.
\begin{itemize}
\item[i)] We have by switching from Stratonovitch to It\^{o} integration and applying  \eqref{eq:Strat-comp}, \eqref{eq:SI-est} and \eqref{eq:prod-ineq} that
\begin{align*}
\biggl\Vert \int_0^\cdot  h_s \, \circ dC^\delta_s \biggr\Vert_{\mathcal{H}^p} &  = \biggl\Vert \int_0^\cdot h_s \, \delta \kappa\bigl(\theta - C^\delta_t\bigr) \, dt + \sqrt{\delta}\gamma \int_0^\cdot h_s  \, \circ dW_t^2 \biggr\Vert_{\mathcal{H}^p} \\& = \biggl\Vert \int_0^\cdot h_s \, \delta \kappa\bigl(\theta - C^\delta_t\bigr) \, dt   + \sqrt{\delta}\gamma \int_0^\cdot h_s  \, dW_t^2 + \frac{1}{2}\sqrt{\delta}\gamma  \int_0^\cdot \beta^2_s \, dt \biggr\Vert_{\mathcal{H}^p}  \\
& \leq   \bigl\vert \theta - c \bigr\vert \kappa \delta \biggl\Vert \int_0^\cdot h_s   e^{-\kappa \delta s} \, ds \biggr\Vert_{\mathcal{H}^p} + \delta \kappa \biggl\Vert \int_0^\cdot h_s  \int_0^s e^{-\kappa \delta (s-u)}\, dW^2_u \, ds \biggr\Vert_{\mathcal{H}^p} \\
&\phantom{==} + \sqrt{\delta} \gamma \biggl\Vert \int_0^\cdot h_s  e^{-\kappa \delta (t-s)}\, dW^2_s \biggr\Vert_{\mathcal{H}^p} + \frac{1}{2}\sqrt{\delta}\gamma \biggl\Vert \int_0^\cdot \beta^2_s  \, ds \biggr\Vert_{\mathcal{H}^p} \\
&\leq   \bigl\vert \theta - c \bigr\vert \kappa \delta \bigl\Vert  h \bigr\Vert_{\mathcal{H}^p}  + 4\bigl(2 + \sqrt{2}\bigr)p\delta \kappa\bigl\Vert  h \bigr\Vert_{\mathcal{H}^{2p}}  + 2\bigl(1+\sqrt{2}\bigr) \sqrt{\delta}\gamma\bigl\Vert  h \bigr\Vert_{\mathcal{H}^p}  + \frac{1}{2}\sqrt{\delta}\gamma \bigl\Vert  h \bigr\Vert_{\mathcal{H}^p},\\
&\leq \sqrt{\delta} \Bigl(\bigl\vert\theta - c\bigr\vert \kappa \sqrt{\delta} + 4\bigl(2 + \sqrt{2}\bigr)\sqrt{\delta} \kappa  + 2\bigl(1+\sqrt{2}\bigr) \gamma + \frac{1}{2}\gamma\Bigr)p \bigl\Vert  h \bigr\Vert_{\mathcal{H}^{2p}}.
\end{align*}
Thus the results holds for  $\delta \leq 1$ with $K_4 := \bigl(\bigl\vert\theta - c\bigr\vert  + 4\bigl(2 + \sqrt{2}\bigr)\bigr) \kappa  + \bigl(2\sqrt{2} + \frac{5}{2}\bigr)\gamma$.
\item[ii)] Using the same line of argument as in part i)
\begin{align*}
\biggl\Vert \int_0^\cdot h_s \, \circ dX^\delta_s \biggr\Vert_{\mathcal{H}^p} & \leq \biggl\Vert \int_0^\cdot h_s  \Bigl(r  - \frac{\sigma^2}{2} - \theta + (\theta -c)e^{-\kappa \delta s} +  \sqrt{\delta} \gamma \int_0^s  e^{-\delta \kappa(s-u)} \, \circ dW^2_u\Bigr) \, ds \biggr. \\
&\phantom{==} +  \int_0^\cdot h_s  \sigma\, \circ dW^1_s\biggr) \biggr\Vert_{\mathcal{H}^p} \\
& = \biggl\Vert \int_0^\cdot h_s  \Bigl(r  - \frac{\sigma^2}{2} - \theta + (\theta -c)e^{-\kappa \delta s} +  \sqrt{\delta} \gamma \int_0^s  e^{-\delta \kappa(s-u)} \, dW^2_u\Bigr) \, ds\biggr. \\
&\phantom{==} \biggl. +  \int_0^\cdot h_s  \sigma\, dW^1_s + \frac{1}{2}\sigma  \int_0^\cdot \beta^1_s  \, ds \biggr) \biggr\Vert_{\mathcal{H}^p} \\
& \leq \Bigl\vert r - \frac{\sigma^2}{2}-\theta\Bigr\vert \bigl\Vert h \bigr\Vert_{\mathcal{H}^{p}} + \bigl\vert \theta-c\bigr\vert^2\bigl\Vert h \bigr\Vert_{\mathcal{H}^{p}} + 24p\sqrt{\delta} \gamma \bigl\Vert h \bigr\Vert_{\mathcal{H}^{2p}} \\
& \phantom{==} + 2\bigl(1+\sqrt{2}\bigr) \sqrt{p} \sigma \bigl\Vert h \bigr\Vert_{\mathcal{H}^p} + \frac{1}{2} \sigma \bigl\Vert h\bigr\Vert_{\mathcal{H}^p}\\
& \leq \biggl(\Bigl\vert r - \frac{\sigma^2}{2} -\theta\Bigr\vert+  \bigl\vert \theta-c\bigr\vert  + 4\bigl(2 + \sqrt{2}\bigr)\sqrt{\delta} \gamma + 2\bigl(1+\sqrt{2}\bigr)  \sigma + \frac{1}{2} \sigma \biggr) p \bigl\Vert h \bigr\Vert_{\mathcal{H}^{2p}}.
\end{align*}
Thus the results holds for  $\delta \leq 1$ with $K_5 := \bigl\vert r - \frac{\sigma^2}{2} -\theta\bigr\vert+  \bigl\vert \theta-c\bigr\vert +  4\bigl(2 + \sqrt{2}\bigr) \gamma + \bigl(2\sqrt{2}+ \frac{5}{2} \bigr)\sigma$.
\item[iii)] We note that
\[
\circ dC_t^\delta = \delta \kappa\bigl(\theta - C^\delta_t\bigr) \, dt + \sqrt{\delta}\gamma \, dW_t^2 = \delta \kappa\Bigl(\bigl(\theta - c\bigr)e^{-\delta \kappa t}- \sqrt{\delta} \gamma \int_0^t e^{-\delta \kappa(t-s)} \, dW^2_s\Bigr) \, dt + \sqrt{\delta}\gamma \, dW_t^2
\]
and
\begin{align*}
\circ dC_t^{(n)} & = \gamma \sqrt{\delta} \sum_{k=1}^{\lfloor \frac{n+1}{2}\rfloor} \delta^{k-1} \frac{(-\kappa)^{k-1}}{(k-1)!} d\Bigl(\int_0^t (t-s)^{k-1} \, dW^2_s\Bigr) + (c-\theta) \sum_{k=1}^{\lfloor \frac{n}{2}\rfloor} \delta^{k}\frac{(-\kappa)^{k}}{(k-1)!} t^{k-1} dt \\
& = \delta \kappa (\theta- c) \sum_{k=0}^{\lfloor \frac{n}{2}\rfloor -1} \delta^{k}\frac{(-\kappa)^{k}}{k!} t^{k} dt + \gamma \sqrt{\delta} \sum_{k=0}^{\lfloor \frac{n+1}{2}\rfloor -2} \delta^{k+1} \frac{(-\kappa)^{k+1}}{k!} \int_0^t (t-s)^{k} \, dW^2_s dt + \sqrt{\delta}\gamma \, dW_t^2.
\end{align*}
Therefore, using \eqref{eq:SI-est} and \eqref{eq:prod-ineq} as well as $\sum_{k=0}^{2n+1}\frac{(-x)^k}{k!} \leq e^{-x} \leq \sum_{k=0}^{2n}\frac{(-x)^k}{k!}$  we get
\begin{align*}
& \phantom{==}\biggl\Vert \int_0^\cdot h_t \, \circ dC^\delta_t - \int_0^\cdot h_t \, \circ dC^{(n)}_t \biggr\Vert_{\mathcal{H}^p} = \biggl\Vert   \int_0^\cdot h_t \, \circ d\bigl(C^\delta_t - C^{(n)}_t \bigr)\biggr\Vert_{\mathcal{H}^p}\\
& \leq \delta \kappa \bigl\vert \theta- c\bigr\vert \biggl\Vert \int_0^\cdot h_t \biggl(e^{-\delta \kappa t} - \sum_{k=0}^{\lfloor \frac{n}{2}\rfloor -1} \delta^{k}\frac{(-\kappa)^{k}}{k!} t^{k}\biggr)\, dt  \biggr\Vert_{\mathcal{H}^p} \\
& \phantom{==} + \gamma \kappa \sqrt{\delta}^3 \biggl\Vert \int_0^\cdot h_t \biggl(\int_0^t e^{-\delta \kappa(t-s)} -  \sum_{k=0}^{\lfloor \frac{n+1}{2}\rfloor -2} \delta^{k} \frac{(-\kappa)^{k}}{k!}(t-s)^{k}  \, dW^2_s\biggr) \, dt  \biggr\Vert_{\mathcal{H}^p} \\
& \leq \delta \kappa \bigl\vert \theta- c\bigr\vert \biggl\vert  e^{-\delta \kappa } - \sum_{k=0}^{\lfloor \frac{n}{2}\rfloor -1} \delta^{k}\frac{(-\kappa)^{k}}{k!}   \biggr\vert \bigl\Vert h \bigr\Vert_{\mathcal{H}^{p}} + 4\bigl(2 + \sqrt{2}\bigr) \gamma \kappa \sqrt{\delta}^3 p \biggl\vert  e^{-\delta \kappa} -  \sum_{k=0}^{\lfloor \frac{n+1}{2}\rfloor -2} \delta^{k} \frac{(-\kappa)^{k}}{k!} \biggr\vert   \bigl\Vert h  \bigr\Vert_{\mathcal{H}^{2p}} \\
& \leq \delta \kappa \bigl\vert \theta- c\bigr\vert  \delta^{\lfloor \frac{n}{2}\rfloor}\frac{\kappa^{\lfloor \frac{n}{2}\rfloor}}{\lfloor \frac{n}{2}\rfloor!}  \bigl\Vert h \bigr\Vert_{\mathcal{H}^{p}}+ 4\bigl(2 + \sqrt{2}\bigr) \gamma \kappa \sqrt{\delta}^3 p \delta^{\lfloor \frac{n+1}{2}\rfloor -1} \frac{\kappa^{\lfloor \frac{n+1}{2}\rfloor -1}}{(\lfloor \frac{n+1}{2}\rfloor -1)!}  \bigl\Vert h  \bigr\Vert_{\mathcal{H}^{2p}}.
\end{align*}
It follows that for $\delta \leq 1$
\[
\biggl\Vert \int_0^\cdot h_t \, \circ dC^\delta_t - \int_0^\cdot h_t \, \circ dC^{(n)}_t \biggr\Vert_{\mathcal{H}^p} \leq \sqrt{\delta}^{n+1} K_6  p \bigl\Vert  h \bigr\Vert_{\mathcal{H}^{2p}}
\]
with $K_6 = \bigl((\kappa \vee 1) \bigl\vert \theta - c \bigr\vert + 4\bigl(2 + \sqrt{2}\bigr) \gamma (\kappa \wedge 1)\bigr)\frac{\kappa^{\kappa+\frac{3}{2}}}{\Gamma(\kappa+1)}$ as  $\max_{z \in [0,\infty)} \frac{a^z}{\Gamma(z+1)} \leq \frac{a^{a+1}}{\Gamma(a+1)}$ for $a>0$.
\item[iv)] Similar to iii) we have here
\begin{align*}
& \phantom{:=} \biggl\Vert \int_0^\cdot h_s \, dX^\delta_s - \int_0^\cdot h_s \, dX^{(n)}_s \biggr\Vert_{\mathcal{H}^p}  = \Bigl\Vert \int_0^\cdot h_s \bigl(C^\delta_s - C^{(n)}_s\bigr) \, ds \Bigr\Vert_{\mathcal{H}^p}\\
& \leq \bigl\vert \theta - c \bigr\vert \Bigl\Vert \int_0^\cdot h_s  \Bigl(e^{- \delta\kappa  t}- \sum_{k=0}^{\lfloor \frac{n}{2}\rfloor} \frac{(-\delta\kappa t)^{k}}{k!} \Bigr) \, ds \Bigr\Vert_{\mathcal{H}^p} \\
& \phantom{==}+  \sqrt{\delta} \gamma \Bigl\Vert\int_0^\cdot  h_s \int_0^s  \biggl(e^{-\delta \kappa(t-s)} -\sum_{k=0}^{\lfloor \frac{n+1}{2}\rfloor-1} \frac{(-\delta\kappa (t-s))^{k}}{k!}\biggr)\, dW^2_u ds  \Bigr\Vert_{\mathcal{H}^p} \\
& \leq \bigl\vert \theta - c \bigr\vert \biggl\vert e^{- \delta\kappa}- \sum_{k=0}^{\lfloor \frac{n}{2}\rfloor} \frac{(-\delta\kappa)^{k}}{k!} \biggr\vert \bigl\Vert h \bigr\Vert_{\mathcal{H}^p} + 4\bigl(2 + \sqrt{2}\bigr) \sqrt{\delta} \gamma p \biggl\vert e^{-\delta \kappa} -\sum_{k=0}^{\lfloor \frac{n+1}{2}\rfloor-1} \frac{(\delta\kappa )^{k}}{k!}\biggr\vert  \bigl\Vert h \Bigr\Vert_{\mathcal{H}^p} \\
& \leq \bigl\vert \theta - c \bigr\vert \frac{(\delta\kappa )^{\lfloor \frac{n}{2}\rfloor +1}}{(\lfloor \frac{n}{2}\rfloor + 1)!} \bigl\Vert h \bigr\Vert_{\mathcal{H}^p} + 4\bigl(2 + \sqrt{2}\bigr) \sqrt{\delta} \gamma p  \frac{(\delta\kappa )^{\lfloor \frac{n+1}{2}\rfloor}}{\lfloor \frac{n+1}{2}\rfloor!}\biggr\vert  \bigl\Vert h \Bigr\Vert_{\mathcal{H}^{2p}}.
\end{align*}
It follows that for $\delta \leq 1$
\[
\biggl\Vert \int_0^\cdot h_t \, \circ dC^\delta_t - \int_0^\cdot h_t \, \circ dC^{(n)}_t \biggr\Vert_{\mathcal{H}^p} \leq \sqrt{\delta}^{n+1} K_6  p \bigl\Vert  h \bigr\Vert_{\mathcal{H}^{2p}}.
\]
\item[v)] Combining the results of i) and iii) we have
\begin{align*}
\biggl\Vert \int_0^\cdot h_s \, \circ dC^{(n)}_s \biggr\Vert_{\mathcal{H}^p} & \leq \biggl\Vert \int_0^\cdot h_s \, \circ dC^\delta_s \biggr\Vert_{\mathcal{H}^p} + \biggl\Vert \int_0^\cdot h_s \, \circ dC^{(n)}_s - \int_0^\cdot h_s \, \circ dC^\delta_s \biggr\Vert_{\mathcal{H}^p} \\
&\leq \sqrt{\delta} K_4 p \bigl\Vert  h \bigr\Vert_{\mathcal{H}^{2p}} + \sqrt{\delta}^{n+1}K_6 p \bigl\Vert  h \bigr\Vert_{\mathcal{H}^{2p}} \leq \sqrt{\delta} K_7 p \bigl\Vert  h \bigr\Vert_{\mathcal{H}^{2p}} 
\end{align*}
where $K_7 := K_4 + K_6$.
\item[vi)] In the same way as in v), combining the results of ii) and iv) we have
\begin{align*}
\biggl\Vert \int_0^t h_s \, \circ dX^{(n)}_s \biggr\Vert_{\mathcal{H}^p} \leq K_5 p \bigl\Vert  h \bigr\Vert_{\mathcal{H}^{2p}} + \sqrt{\delta}^{n+1}K_6 p \bigl\Vert  h \bigr\Vert_{\mathcal{H}^{2p}} \leq K_8 p \bigl\Vert  h \bigr\Vert_{\mathcal{H}^{2p}} 
\end{align*}
for $\delta \leq 1$ where $K_8 := K_5 + K_6$. 
\end{itemize}
\end{proof}

\end{document}